\documentclass[12pt]{article}   
\usepackage[english]{babel}
\usepackage[toc,page]{appendix}
\usepackage{bbm}
\usepackage{fullpage}
\usepackage[top=1in, bottom=1.25in, left=1in, right=1in]{geometry}
\usepackage{etaremune}
\usepackage{enumitem}

\usepackage{amssymb,amsmath,amsthm}

\usepackage{hyperref}

\hypersetup{
  colorlinks   = true, 
  urlcolor     = blue, 
  linkcolor    = blue, 
  citecolor   = red 
}

\usepackage[labelformat=empty]{caption}

\newtheorem{theorem}{Theorem}[section]
\newtheorem{lemma}[theorem]{Lemma}

\newtheorem{proposition}[theorem]{Proposition}
\newtheorem{corollary}[theorem]{Corollary}

\theoremstyle{definition}
\newtheorem{definition}{Definition}
\newtheorem{remark}[theorem]{Remark}

\numberwithin{equation}{section}

\usepackage[suppress]{color-edits}  
\addauthor{rev}{blue}   
\addauthor{ms}{red}

\usepackage{times}

\newcommand{\bea}{\begin{eqnarray}}
\newcommand{\eea}{\end{eqnarray}}
\newcommand{\<}{\langle}
\renewcommand{\>}{\rangle}

\newcommand{\wt}{\widetilde}

\newcommand{\wh}{\widehat}

\def\eps{{\varepsilon}}

\def\supp{{\rm supp}}

\def\bP{{\boldsymbol{P}}}

\def\bQ{{\boldsymbol{Q}}}

\def\bone{{\mathbf 1}}
\def\cF{{\mathcal F}}

\def\good{\footnotesize\textbf{\texttt{good}}}
\def\bad{\footnotesize\textbf{\texttt{bad}}}

\def\bB{\boldsymbol{B}}

\def\de{{\rm d}}

\def\<{\langle}
\def\>{\rangle}

\def\b0{{\boldsymbol{0}}}

\renewcommand{\b}{\mathbf{b}}

\newcommand{\eff}{\text{eff}}

\def\lt{\left}
\def\rt{\right}

\def\la{\langle}
\def\ra{\rangle}

\def\eps{\varepsilon}

\def\bbB{{\mathbb{B}}}
\def\bbE{{\mathbb{E}}}

\def\bbP{{\mathbb{P}}}
\def\bbQ{{\mathbb{Q}}}
\def\bP{\boldsymbol{P}}
\def\bbR{{\mathbb{R}}}

\def\bbZ{{\mathbb{Z}}}

\def\cF{{\mathcal{F}}}

\def\one{{\mathbbm{1}}}

\def\bP{\mathbf{P}}
\def\bQ{\mathbf{Q}}

\def\TV{{\mathrm{TV}}}

\author{Mark Sellke}

\title{Almost Quartic Lower Bound for the Fr\"{o}hlich Polaron's Effective Mass via Gaussian Domination}

\date{}

\begin{document}

\maketitle

\begin{abstract}
\noindent We prove the Fr\"{o}hlich polaron has effective mass at least $\frac{\alpha^4}{(\log \alpha)^6}$ when the coupling strength $\alpha$ is large. This nearly matches the quartic growth rate $C_*\alpha^4$ predicted in \cite{landau1948effective} and complements a recent sharp upper bound of \cite{brooks2022fr}. Our proof works with the path integral formulation of the problem and systematically applies the Gaussian correlation inequality to exploit quasi-concavity of the interaction terms. 
\end{abstract}


\section{Introduction}

The Fr\"{o}hlich polaron in quantum mechanics was introduced in \cite{frohlich1937theory} to model an electron or other charged particle in a crystal. As the electron moves, it drags along a surrounding polarization cloud, and these together form a polaron. In this paper we obtain new estimates for the \emph{effective mass} of a polaron using the path integral description of \cite{feynman1955slow}.

Before giving the path integral formulation, we briefly review the original quantum mechanical model. Here the polaron at coupling strength $\alpha\geq 0$ is described by a Hamiltonian operator $H$ on $L^2(\bbR^3)\otimes \cF(L^2(\bbR^3))$, where the latter term is a bosonic Fock space. With $x$ lying in the former $\bbR^3$ space and $k$ the latter, and for $a_k^{\dagger},a_k$ the creation and annihilation operators, $H$ is given by
\[
    H=-\nabla_x^2/2 + \int_{\bbR^3} a_k^{\dagger}a_k~\de k + 
    \sqrt{\alpha}\int_{\bbR^3} \frac{e^{-ikx}}{|k|} a_k^{\dagger}~\de k
    +
    \sqrt{\alpha}\int_{\bbR^3} \frac{e^{ikx}}{|k|}a_k~\de k.
\]
Because $H$ commutes with the total momentum operator $-i\nabla_x+\int_{\bbR^3} ka_k^{\dagger}a_k~\de k$ and is rotationally invariant, it has a well-defined ground state energy $E(P)=E_{\text{rad}}(|P|)$ for each momentum $P\in\bbR^3$. It is known since \cite{gross1972existence} that $E(0)\leq E(P)$ for all $P$, and more recently from \cite{polzer2022renewal} that $E_{\text{rad}}$ is monotone increasing and strictly so at $0$. Along with the overall ground state energy $E(0)$, one of the main quantities of interest is the effective mass $m_{\eff}(\alpha)$ defined by
\begin{equation}
\label{eq:meff-QFT}
    \frac{1}{2m_{\eff}(\alpha)}=\lim_{P\to 0} \frac{E(P)-E(0)}{|P|^2}.
\end{equation}
See also \cite{lieb2014equivalence} for an equivalent ``static'' definition of $m_{\eff}(\alpha)$ involving potential wells.

We now turn to the probabilistic path integral description which originated in \cite{feynman1955slow}. Let $\bbP$ be the law of $3$-dimensional Brownian motion and fix a coupling strength $\alpha$ 
and time horizon $T>0$. Then the \emph{polaron path measure} $\wh\bbP_{\alpha,T}$ on $C([0,T];\bbR^3)$ is given by:
\begin{align}
\label{eq:polaron-defn}
    \de\wh\bbP_{\alpha,T}(\bB)
    &\equiv
    \frac{1}{Z_{\alpha,T}}
    \exp\lt(
    \alpha\int_0^T \int_0^T
    e^{-|t-s|} V(\|\bB_t-\bB_s\|) ~\de t~\de s
    \rt)
    \de\bbP(\bB),
    \\
\nonumber
    V(r)&\equiv1/r
    .
\end{align}
Although $V(\cdot)$ is singular, finiteness of $Z_{\alpha,T}$ is proved in e.g. \cite{bley2017estimates}. Later \cite{spohn1987effective,dybalski2020effective} showed that assuming a functional central limit theorem for $\wh\bbP_{\alpha,T}$ with $L^2$ convergence at large times $T$, the original definition \eqref{eq:meff-QFT} of $m_{\eff}(\alpha)$ coincides with:
\begin{equation}
\label{eq:sigma-def}
\begin{aligned}
    m_{\eff}(\alpha)&=1/\sigma^{2}(\alpha);
    \\
    \sigma^{2}(\alpha)
    &\equiv
    \lim_{T \to \infty} 
    \frac
    {
    \bbE^{ \wh\bbP_{\alpha,T}}
    \left[\|\bB_{T}\|^2\right]
    }
    {3T}\,.
\end{aligned}
\end{equation}
The required functional central limit theorem for $\wh\bbP_{\alpha,T}$ was subsequently shown in \cite{mukherjee2018identification,betz2022functional}, who in fact represented $\wh\bbP_{\alpha,T}$ as a mixture of Gaussians. We note that the path integral description was also used much earlier by \cite{donsker1983asymptotics} to compute the ground state energy $E(0)$. In particular its $\alpha\to \infty$ limit was expressed as an explicit Pekar variational problem analyzed in \cite{lieb1977existence}. Later \cite{lieb1997exact} gave a non-probabilistic proof for this limit, and \cite{frank2021quantum,feliciangeli2021strongly,brooks2022v1fr} recently determined the second order ``quantum correction'' term. In the ``strong coupling'' regime $\lim_{T\to\infty}\lim_{\alpha\to\infty}$, the short-time pathwise behavior of $\wh\bbP_{\alpha,T}$ was shown in \cite{bolthausen2017mean,mukherjee2020identification} to be described by a stationary diffusion known as the Pekar process.

We will focus on the asymptotic growth of $m_{\eff}(\alpha)$ as $\alpha\to\infty$. Precise predictions have long been known in the physics literature since \cite{landau1948effective}, who conjectured the quartic behavior 
\revedit{$\lim_{\alpha\to\infty} m_{\eff}(\alpha)/\alpha^4= C_*$, even supplying a explicit $C_*$ given by Pekar's variational formula}. However no nontrivial bounds on $m_{\eff}(\alpha)$ were proved until much more recently. The divergence $\lim_{\alpha\to\infty}m_{\eff}(\alpha)=\infty$ was first shown non-probabilistically in \cite{lieb2020divergence} and later improved to $m_{\eff}(\alpha)\geq c\alpha^{2/5}$ in \cite{betz2022effective} using the mixture-of-Gaussians description from \cite{mukherjee2018identification,betz2022functional}. (See also \cite[Theorem 2]{mysliwy2022polaron} for a rigorous lower bound in some other polaron models.)
Even more recently \cite{brooks2022fr} proved the sharp \emph{upper bound} $m_{\eff}(\alpha)\leq C_*\alpha^4+O(\alpha^{4-\eps})$. Their work uses operator techniques to study $E(P)$ for not-too-small momentums $|P|\geq \alpha^{-1-\eps}$, which suffices thanks to a concavity result of \cite{polzer2022renewal}. 

Our main result establishes the matching quartic \emph{lower bound} on $m_{\eff}(\alpha)$ up to $\log(\alpha)$ factors. In tandem with the upper bound of \cite{brooks2022fr}, this almost resolves the conjecture of \cite{landau1948effective}.

\begin{theorem}
\label{thm:main}
    For $\alpha\geq 2$ and an absolute constant $c$,
    \[
    m_{\eff}(\alpha)\geq \frac{c\alpha^4}{(\log\alpha)^6}\,.
    \]
    In fact, uniformly over $\alpha\geq 2$ and $T\geq 1$ we have
    \[
    \bbE^{ \wh\bbP_{\alpha,T}}
    \left[\|\bB_{T}\|^2\right]
    \leq
    O\lt(\frac{T(\log\alpha)^6}{\alpha^4}+\frac{(\log\alpha)^3}{\alpha^2}\rt).
    \]
\end{theorem}

Our approach stems from the following observation. Although the Radon--Nikodym derivative $\de \wh\bbP_{\alpha,T}/\de \bbP$ fails to be log-concave, it is nonetheless approximated (via Riemann sums for the double integral in the exponent of \eqref{eq:polaron-defn}) by finite products of origin-symmetric functions of $\bB_{[0,T]}$. Moreover each of these functions is \emph{quasi-concave} because $V(r)=1/r$ is decreasing (ignoring for the moment the singularity at $V(0)$). The strategy of this paper is thus to upper bound $\sigma^2(\alpha)$ by systematically applying Royen's Gaussian correlation inequality. Indeed an infinite-dimensional, functional version of this inequality ought to imply $\wh\bbP_{\alpha,T}$ is dominated by $\bbP$ in that
\[
    \wh\bbP_{\alpha,T}[K]\geq \bbP[K]
\]
for any symmetric convex set $K\subseteq C([0,T];\bbR^3)$.

This domination by Brownian motion was previously used in \cite{betz2022functional} as a tightness condition toward proving a functional central limit theorem. But for the effective mass itself, it implies only that $\sigma^2(\alpha)\leq 1$, i.e. $m_{\eff}(\alpha)\geq 1$. To see how to progress further, let us imagine that $\bbP$-almost surely, the bound $\|\bB_s-\bB_t\|\leq R$ holds for all $|t-s|\leq 1$. 
\revedit{It is easily seen that $V(r)+\frac{r^2}{2R^3}$ is decreasing on $r\in [0,R]$; hence the Gaussian correlation inequality would show domination of $\wh\bbP_{\alpha,T}$ by a ``more confined'' Gaussian measure in which $\bbP$ is weighted by 
\[
    \exp
    \lt(
    -\frac{\alpha}{2\text{e}R^3}
    \int_0^T\int_0^T \|\bB_t-\bB_s\|^2 \cdot \one_{|t-s|\leq 1}
    ~\de t~\de s
    \rt).
\]
}
Analyzing this reweighted measure would then yield non-trivial lower bounds for $m_{\eff}(\alpha)$. Although $\|\bB_s-\bB_t\|$ is bounded only with high probability, our argument is built around a more precise version of this idea.

\begin{remark}
\label{rem:robust}
Our proof of Theorem~\ref{thm:main} certainly requires the base measure $\bbP$ to be Gaussian, and might become more involved if $\bbP$ did not have independent increments. However it is somewhat robust to changes in the interaction term. Changing the spatial dimension to be different from $3$ affects neither the argument nor the bounds obtained (so long as $Z_{\alpha,T}<\infty$). Moreover the factor $e^{-|t-s|}$ in \eqref{eq:polaron-defn} can be replaced by any non-negative continuous bounded function of $(t,s)$ which is uniformly positive for $|t-s|\leq c$ \revedit{(and decays fast enough that $Z_{\alpha,T}$ is finite)}. Likewise the potential $V(r)=1/r$ can be replaced by $V_{s,t}(\|\bB_t-\bB_s\|)$ for any decreasing differentiable functions $V_{s,t}:(0,\infty)\to\bbR$ such that $V'_{s,t}(r)\leq -C/r^2$ for some uniform constant $C$ and all $r,s,t>0$, again so long as $Z_{\alpha,T}<\infty$. 
Even taking a different power law $V(r)=r^{-p}$ for $0<p<2$ only affects the proof numerically, and in Remark~\ref{rem:effective-mass-general-p} we derive in this case the bound
\begin{equation}
\label{eq:effective-mass-general-p}
    m_{\eff}^{(p)}(\alpha)\geq 
    \frac{c_p\, \alpha^{\frac{4}{2-p}}}
    {(\log\alpha)^{\frac{4+2p}{2-p}}}.
\end{equation}
The exponent $\frac{4}{2-p}$ is consistent with a natural generalization \cite[Equation (1.30)]{mysliwy2022polaron} of the Pekar conjecture, so it is likely sharp for all $0<p<2$. Note that \cite{bley2017estimates} shows finiteness of $Z_{\alpha,T}$ for $p<2$. If $p>2$ then $Z_{\alpha,T}=\infty$ by considering the events $\{\sup_{t\in [0,1]} \|\bB_t\|\leq \eps\}$ for small $\eps$.
\end{remark}

\section{Preliminaries}
\label{sec:prelims}

In this section, after establishing some basic notations we review the crucially important Gaussian correlation inequality. Then we explain how to truncate the potential $V$ and discretize time, and then outline the idea for our main proof. We end the section by highlighting some notational conventions for various path measures.

\subsection{Basic Notation}

We always use $\mu$ to denote a centered Gaussian measure on a finite-dimensional real vector space $X$. Throughout we say a convex set $K\subseteq X$ is \emph{symmetric} if $K=-K$ is invariant under negation, and similarly define symmetric probability measures and functions on $X$. We write $\de\nu_1/\de\nu_2$ for the Radon--Nikodym derivative, and define $\mu^{\times 2}$ to be the dilated Gaussian measure with 
\begin{equation}
\label{eq:mu-times-2}
    \mu^{\times 2}(2A)=\mu(A)
\end{equation}
for all Borel $A\subseteq X$.

\revedit{We write $[k]=\{0,1,2\dots,k-1\}$ for $k$ a positive integer. We write $F\leq O(G)$ or $F=O(G)$ for non-negative $F,G$ to indicate that $F/G$ is at most an absolute constant, independent of any other quantities (e.g. $\alpha,T$ as well as $A,\eta$ defined later). We often use $C$ to denote such absolute constants. We write $F=\Theta(G)$ or equivalently $F\asymp G$ if $F\leq O(G)$ and $G\leq O(F)$. We use the notation $\propto$ to indicate equivalence of positive measures up to normalizing constants, and in particular often use it to \emph{define} a probability measure as in the following Definition~\ref{def:reweight}.
}

\begin{definition}
\label{def:reweight}
    For a probability measure $\nu$ on $X$ and non-negative function $f$ with $\bbE^{\nu}[f]\in (0,\infty)$, define the reweighting
    \[
    \de\nu^{(f)}(x)\propto f(x)\de\nu(x),
    \]
    given explicitly by
    \[
    \nu^{(f)}(A)=\frac{\int_A f(x)\de \nu(x)}{\int_X f(x)\de\nu(x)}.
    \]
    If $Q:X\to \bbR_{\geq 0}$ is a symmetric non-negative quadratic function, define
    \[
    \de\nu^{\la Q\ra}(x)\propto e^{-Q(x)}\de\nu(x).
    \]
\end{definition}

\subsection{Gaussian Correlation Inequality}

We say $f:X\to\bbR_{\geq 0}$ is \emph{quasi-concave} if the super level sets $\{x\in X~:~f(x)\geq C\}$ are convex for all $C\in\bbR$. We now recall Royen's Gaussian correlation inequality, proved in the remarkable paper \cite{royen2014simple}, as well as its functional form.

\begin{theorem}
\label{thm:GCI}
    Let $\mu$ be a centered Gaussian measure on $X$. Then the following statements hold:
    \begin{enumerate}
    \item For any symmetric convex sets $K_1,K_2\subseteq X$, 
    \begin{equation}
    \label{eq:GCI-1}
    \mu(K_1)\mu(K_2)\leq \mu(K_1\cap K_2).
    \end{equation}
    \item For any symmetric convex sets $K_1,K_2,\dots,K_n\subseteq X$,
    \begin{equation}
    \label{eq:GCI-2}
    \mu\lt(\bigcap_{j=1}^m K_j\rt)\cdot \mu\lt(\bigcap_{k=m+1}^n K_k\rt)
    \leq 
    \mu\lt(\bigcap_{i=1}^n K_i\rt).
    \end{equation}
    \item For any symmetric quasi-concave functions $f_1,f_2,\dots,f_n:X\to\mathbb R_{\geq 0}$,
    \begin{equation}
    \label{eq:GCI-3}
    \bbE^{x\sim \mu}\lt[\prod_{j=1}^m f_j(x)\rt]\cdot
    \bbE^{x\sim \mu}\lt[\prod_{k=m+1}^{n} f_k(x)\rt]
    \leq
    \bbE^{x\sim \mu}\lt[\prod_{i=1}^n f_i(x)\rt].
    \end{equation}
    \end{enumerate}
\end{theorem}

\begin{proof}
    First, \eqref{eq:GCI-1} is the usual statement of the Gaussian correlation inequality as proved in \cite{royen2014simple} (see also the exposition \cite{latala2017royen}). Next, \eqref{eq:GCI-2} follows trivially from \eqref{eq:GCI-1} since the intersection of symmetric convex sets is again symmetric and convex. Finally \eqref{eq:GCI-3} follows from \eqref{eq:GCI-2} via Fubini and multi-linearity: each $f_i$ is a positive combination of indicators of symmetric convex sets, namely its own super-level sets. 
\end{proof}

\begin{definition}
    For symmetric probability measures $\mu,\nu$ on $X\simeq \bbR^n$, we write
    \[
    \nu\preceq \mu
    \]
    if $\de\nu/\de\mu$ is a finite product of symmetric quasi-concave functions, or a uniformly bounded a.e. limit of such.
\end{definition}

\begin{corollary}
\label{cor:GCI}
    Let $\mu,\nu$ be symmetric probability measures on $X\simeq \bbR^n$ with $\mu$ Gaussian and $\nu\preceq\mu$. Then $\nu(K)\geq \mu(K)$ for any symmetric convex set $K$ and $\bbE^{\nu}[f]\leq \bbE^{\mu}[f]$ for any \revedit{non-negative symmetric} convex function $f$.
\end{corollary}

Note that the relation $\preceq$ is transitive, and is preserved by reweighting $\nu$ by a quasi-concave function, or by reweighting both $\nu$ and $\mu$ by \revedit{the same factor. We will often use the latter fact with weight factor} $\exp(-Q)$ for a non-negative symmetric quadratic $Q$. Finally we record an obvious but useful special case.

\begin{corollary}
\label{cor:quadratic-domination}
    Let $\mu$ be a centered Gaussian measure on $X$, and $Q:X\to \bbR_{\geq 0}$ a symmetric non-negative quadratic function. Then $\mu^{\la Q\ra}$ is a Gaussian measure and $\mu^{\la Q\ra}\preceq \mu$. 
\end{corollary}

\subsection{Cutting Off the Interaction and Finite Dimensional Approximation}

To remove technical issues, we will truncate the polaron interaction to be bounded and also discretize time. First, we define the cut-off potential
\begin{equation}
\label{eq:VA-defn}
    V_A(r)=
    \begin{cases}
    2A-A^2 r,\quad r\in [0,1/A]\\
    \frac{1}{r},\quad r\in [1/A,\infty).
    \end{cases}
\end{equation}
Thus $V_A$ is a uniformly bounded approximation to $V(r)=1/r$ which is still decreasing and convex. Note also that $V_A$ increases pointwise up to $V$ as $A\to\infty$. 

\begin{proposition}
\label{prop:gaussian-domination-computation}
    For any $A>1/R>0$, the function $V_A(r)+\frac{r^2}{2R^3}$ is decreasing on $[0,R]$.
\end{proposition}

\begin{proof}
It suffices to note that $V_A(r)+\frac{r^2}{2R^3}$ is convex with derivative vanishing at $r=R$.
\end{proof}

Next, define the cut-off path measure $\wh\bbP_{\alpha,T}^{(A)}$ and weight function $W_{\alpha,T}^{(A)}:C([0,T];\bbR^3)\to \bbR$ by
\begin{align}
\label{eq:truncated-polaron}
    \de\wh\bbP_{\alpha,T}^{(A)}(\bB)
    &=
    \frac{W_{\alpha,T}^{(A)}(\bB) \de\bbP(\bB)}{Z_{\alpha,T}^{(A)}}
   ;
    \\
\label{eq:W-def}
    W_{\alpha,T}^{(A)}(\bB)
    &=
    \exp\lt(\alpha \int_0^T \int_0^T
    e^{-|t-s|} V_A(\|\bB_{s}-\bB_{t}\|)\rt).
\end{align}

\begin{proposition}
\label{prop:A-to-infty}
    For any $T>0$, 
    \[
    \lim_{A\to\infty}
    \bbE^{\wh\bbP_{\alpha,T}^{(A)}}[\|\bB_T\|^2]
    =
    \bbE^{\wh\bbP_{\alpha,T}}[\|\bB_T\|^2].
    \]
\end{proposition}

\begin{proof}
    Given the finiteness of $Z_{\alpha,T}$ as shown in \cite{bley2017estimates}, it suffices to verify the identities
    \begin{align*}
    &\lim_{A\to\infty}
    \bbE^{\bbP}
    \lt[
    \exp\lt(
    \alpha  \int_0^T \int_0^T
    e^{-|t-s|} V_A(\|\bB_t-\bB_s\|) ~\de t~\de s
    \rt)
    \rt]
    \\
    &\quad\quad
    =
    \bbE^{\bbP}
    \lt[
    \exp\lt(
    \alpha\int_0^T \int_0^T
    e^{-|t-s|} V(\|\bB_t-\bB_s\|) ~\de t~\de s
    \rt)
    \rt],
    \\
    &\lim_{A\to\infty}
    \bbE^{\bbP}
    \lt[
    \|\bB_T\|^2
    \exp\lt(
    \alpha  \int_0^T \int_0^T
    e^{-|t-s|} V_A(\|\bB_t-\bB_s\|) ~\de t~\de s
    \rt)
    \rt]
    \\
    &\quad\quad
    =
    \bbE^{\bbP}
    \lt[
    \|\bB_T\|^2
    \exp\lt(
    \alpha\int_0^T \int_0^T
    e^{-|t-s|} V(\|\bB_t-\bB_s\|) ~\de t~\de s
    \rt)
    \rt]
    .
    \end{align*}
    Both follow immediately from the monotone convergence theorem.
\end{proof}

We next define a time discretization into $\eta$-increments, where $\eta^{-1}\in\bbZ_+$ is always assumed. Let $\bbP^{(\eta)}_{[0,T]}$ be the law of the piecewise-linear process which agrees with $\bB_t$ at each time $t\in \eta\bbZ$, for $\bB\sim \bbP_{[0,T]}$ a Brownian motion. We denote by $C^{(\eta)}([0,T];\bbR^3)$ its support, which consists of piecewise-linear functions on $\eta$-intervals. Similarly to before, let
\[
\wh\bbP_{\alpha,T}^{(A,\eta)}=\lt(\bbP^{(\eta)}_{[0,T]}\rt)^{(W_{\alpha,T}^{(A)})}
\]
be the reweighting of $\bbP^{(\eta)}_{[0,T]}$ by $W_{\alpha,T}^{(A)}$.

\begin{proposition}
\label{prop:discrete-benign}
    Let $f:C([0,T];\bbR^3)\to \bbR_{\geq 0}$ be a \revedit{non-negative} continuous and bounded function which is not identically zero. Then 
    \begin{equation}
    \label{eq:continuous-mapping}
    \lim_{\eta\to 0}~
    (\bbP_{[0,T]}^{(\eta)})^{(f)}
    =
    (\bbP_{[0,T]})^{(f)}
    \end{equation}
    holds as probability measures on $C([0,T];\bbR^3)$. \revedit{In addition}, for any $T,A$ we have
    \revedit{
    \begin{equation}
    \label{eq:delta-safe}
    \lim_{\eta\to 0}
    \sup_{t\in [0,T]}
    \big|\bbE^{\wh\bbP_{\alpha,T}^{(A,\eta)}}[\|\bB_t\|^2]
    -
    \bbE^{\wh\bbP_{\alpha,T}^{(A)}}[\|\bB_t\|^2]
    \big|
    =0.
    \end{equation}
    }
\end{proposition}

\begin{proof}
    The case $f\equiv 1$ is well known. To obtain the general case of \eqref{eq:continuous-mapping}, let $g:C([0,T];\bbR^3)\to \bbR$ be bounded and continuous. Then the continuous mapping theorem yields
    \[
    \lim_{\eta\to 0}
    \bbE^{\bbP^{(\eta)}_{[0,T]}}
    \lt[
    f(\bB)g(\bB)
    \rt]
    =
    \bbE^{\bbP_{[0,T]}}
    \lt[
    f(\bB)g(\bB)
    \rt]
    \]
    which easily implies \eqref{eq:continuous-mapping}. 

    \revedit{For the second part, \eqref{eq:continuous-mapping} implies the result with $\|\bB_t\|^2$ replaced by any bounded continuous function of $\bB_{[0,T]}$, i.e. one has weak convergence of $\|\bB_t\|^2$.
    Moreover note that $\bB_t$ has uniformly sub-Gaussian tails for each $\bbP_{[0,T]}^{(\eta)}$ as $t\in [0,T]$ and $\eta\in [0,1]$ vary.
    As $W_{\alpha,T}^{(A)}$ is uniformly bounded for fixed $T$ we deduce that the expectations $\bbE^{\wh\bbP_{\alpha,T}^{(A,\eta)}}[\|\bB_t\|^4]$ are uniformly bounded across $t\in [0,T]$ and $\eta\in (0,1]$. 
    For the same reason, $\bbE^{\wh\bbP_{\alpha,T}^{(A)}}[\|\bB_t\|^4]<\infty$.
    Hence \eqref{eq:delta-safe} follows.}
\end{proof}

In light of Propositions~\ref{prop:A-to-infty} and \ref{prop:discrete-benign}, Theorem~\ref{thm:main} will follow if we prove that for all $A\geq\alpha$,
\begin{equation}
\label{eq:desired-bound-after-simplification}
    \lim_{\eta\to 0}
    \bbE^{\wh\bbP_{\alpha,T}^{(A,\eta)}}[\|\bB_T\|^2]
    \leq
    \frac{C(\log\alpha)^6) T}{\alpha^4}+\frac{C(\log\alpha)^3}{\alpha^2}.
\end{equation}
Establishing the bound \eqref{eq:desired-bound-after-simplification} will be our main goal for the remainder of the paper.

\begin{remark} 
It should be possible to implement our arguments directly in continuous time. However this creates technical complications due to the important Lemma~\ref{lem:good-decomposition}, which constructs a measure $\nu_{\bad_1}$ absolutely continuous with respect to $\mu$ such that $\nu_{\bad_1}\preceq\mu^{\times 2}$. It is not obvious how to make sense of this in infinite dimensions since $\mu$ and $\mu^{\times 2}$ become singular. We believe $\preceq$ can be suitably defined based on finite-dimensional marginals (for example \cite[Proposition 4.2.6]{bogachev1998gaussian} shows how to approximate infinite-dimensional convex sets by convex cylinder sets in an abstract Wiener space). However this route requires reproving many basic properties that are obvious in finite dimension. On the other hand as we have just seen, the continuous mapping theorem suffices as an easy bridge between discrete and continuous time.
\end{remark}

The interaction $V_A(\cdot)$ appears in our main argument via Proposition~\ref{prop:gaussian-domination-general} below. Below we say a function $\bbR^2\to\bbR$ is $\bbZ^2$-piecewise-constant if it is 
constant on each $[i,i+1)\times [j,j+1)$.

\begin{proposition}
\label{prop:gaussian-domination-general}
    Fix $A>0$. Let $\bbQ^{\dagger}$ and $\bbQ$ be probability measures on $C^{(\eta)}([0,T];\bbR^3)$ such that $\bbQ$ is Gaussian and $\bbQ^{\dagger}\preceq\bbQ$. Suppose that for a $\bbZ^2$-piecewise-constant function $R:\bbR_{\geq 0}^2\to [1/A,\infty)\cup\{+\infty\}$, the bound $\|\bB_s-\bB_t\|\leq R(s,t)$ holds $\bbQ^{\dagger}$-almost surely.
    Define $\wh\bbQ$ and $\wt\bbQ$ by
    \begin{align*}
    \de\wh\bbQ(\bB)
    &\propto
    \exp\lt(
    \alpha\int_0^T \int_0^T
    e^{-|t-s|} V_A(\|\bB_t-\bB_s\|) ~\de t~\de s
    \rt)
    \de\bbQ^{\dagger}(\bB)\,;
    \\
    \de\wt\bbQ(\bB)
    &\propto
    \exp\lt(
    -\int_0^T \int_0^T
    F(s,t)\|\bB_s-\bB_t\|^2 ~\de t~\de s
    \rt)
    \de\bbQ(\bB)
    \end{align*}
    for a $\bbZ^2$-piecewise-constant function $F$ satisfying
    \begin{equation}
    \label{eq:F-bound}
    0\leq F(s,t)\leq \frac{e^{-|t-s|} \alpha}{2R(s,t)^3}\quad\forall s,t\in\bbR.
    \end{equation}
    Then $\wt\bbQ$ is a Gaussian measure and
    \[
    \wh\bbQ\preceq \wt\bbQ\preceq \bbQ.
    \]
\end{proposition}

\begin{proof}
Corollary~\ref{cor:quadratic-domination} implies that $\wt\bbQ$ is a Gaussian measure and $\wt\bbQ\preceq \bbQ$. To show $\wh\bbQ\preceq \wt\bbQ$, note that the Radon-Nikodym derivative is by definition
\begin{equation}
\label{eq:RN-derivative-to-Riemann-sum}
    \frac{\de\wh\bbQ}{\de \wt\bbQ}
    =
    C\cdot
    \exp\lt(
    \int_0^T \int_0^T
    \alpha
    e^{-|t-s|}V_A(\|\bB_s-\bB_t\|)
    +
    F(s,t)\|\bB_s-\bB_t\|^2 ~\de t~\de s
    \rt)\cdot
     \frac{\de\bbQ^{\dagger}}{\de\bbQ}
    .
\end{equation}
Proposition~\ref{prop:gaussian-domination-computation} and the assumption~\eqref{eq:F-bound} together imply that 
\[
    r\mapsto 
    \alpha e^{-|t-s|}V_A(r)
    +
    F(s,t)r^2
\]
is decreasing for $r\in [0,R(s,t)]$, hence agrees on this set with a uniformly bounded, decreasing function with domain $[0,\infty)$. Moreover the integral 
\[
    \int_0^T \int_0^T
    \alpha 
    e^{-|t-s|}V_A(\|\bB_s-\bB_t\|)
    +
    F(s,t)\|\bB_s-\bB_t\|^2 ~\de t~\de s
\]
is approximated in measure (via Riemann summation) by uniformly bounded finite sums of symmetric quasi-concave functions on $C^{(\eta)}([0,T];\bbR^3)$. (The $\bbZ^2$-piecewise-constant conditions and the boundedness and continuity of $V_A$ ensure there are no difficulties in this approximation.) 
\revedit{Because the exponential of any symmetric quasi-concave function is again symmetric quasi-concave, it follows that the exponential weight factor in} \eqref{eq:RN-derivative-to-Riemann-sum} is approximated by finite \emph{products} of symmetric quasi-concave functions. This completes the proof.
\end{proof}

\subsection{Proof Outline}
\label{subsec:outline}

We now present a summary of our main argument. We begin by outlining a simplified proof which yields the weaker bound
\revedit{ 
\[
\bbE^{\wh\bbP_{\alpha,T}^{(A,\eta)}}[\|\bB_T\|^2/T]\leq\frac{(\log(\alpha T))^C}{\alpha^2}.
\]
}
Then we briefly discuss some of the arguments needed to remove the factors of $\log T$ and improve the exponent of $\alpha$.

First, we will always assume $T\in\bbZ_+$ for convenience; if not, so long as $T\geq 1$ we can rescale time slightly to make it so. We also fix $A\geq\alpha$. 
\revedit{
Define $K_{R,[0,T]}^{(\eta)}\subseteq C^{(\eta)}([0,1];\bbR^3)$ to be the symmetric convex set
\[
    K_{R,[0,T]}^{(\eta)}=\lt\{\bB\in C^{(\eta)}([0,T];\bbR^3)~:~\sup_{0\leq i\leq T-1}\sup_{t,s\in [i,i+1]}\|\bB_t-\bB_s\|\leq R\rt\}.
\]
}
Then for \revedit{$R=10\sqrt{\log(\alpha T)}$}, Proposition~\ref{prop:gaussian-domination-general} implies that $\wh\bbP_{\alpha,T}^{(A,\eta)}\preceq \bbP^{(\eta)}_{[0,T]}$ and so \revedit{an easy union bound over integers $0\leq i\leq T-1$ implies}
\begin{equation}
\label{eq:confinement-basic-example}
    \wh\bbP_{\alpha,T}^{(A,\eta)}[K_{R,[0,T]}^{(\eta)}]
    \geq 
    \bbP^{(\eta)}_{[0,T]}[K_{R,[0,T]}^{(\eta)}]
    \geq
    1-\alpha^{-10}.
\end{equation}
Let $\wh \bbP_{\alpha,T,R}^{(A,\eta)}$ be the law of $\wh\bbP_{\alpha,T}^{(A,\eta)}$ conditioned to lie in $K_{R,[0,T]}^{(\eta)}$. Then we have the total variation bound
\begin{equation}
\label{eq:tv-bound-simple}
    \|\wh \bbP_{\alpha,T,R}^{(A,\eta)}-\wh\bbP_{\alpha,T}^{(A,\eta)}\|_{\TV}
    \leq \alpha^{-10}
\end{equation}
and by definition
\begin{equation}
\label{eq:GCI-simple-example}
    \wh \bbP_{\alpha,T,R}^{(A,\eta)}
    \preceq
    \wh\bbP_{\alpha,T}^{(A,\eta)}.
\end{equation}
Proposition~\ref{prop:gaussian-domination-general} then implies
\begin{equation}
\label{eq:comparison-example}
    \wh\bbP_{\alpha,T,R}^{(A,\eta)}\preceq \wt\bbP^{(\eta)}_{[0,T]}
\end{equation}
where $\wt\bbP^{(\eta)}_{[0,T]}$ is defined by
\begin{equation}
\label{eq:confinement-example}
    \de\wt\bbP_{[0,T]}^{(\eta)}(\bB)
    \propto 
    \exp\lt(
    -\frac{\alpha}{2\text{e}R^3}
    \sum_{i=0}^{T-1}
    \int_{i}^{i+1}
    \int_{i}^{i+1}
    \|\bB_t-\bB_s\|^2 \de t\de s
    \rt) \de\bbP_{[0,T]}^{(\eta)}(\bB)
    .
\end{equation}

The measure $\wt\bbP_{[0,T]}^{(\eta)}$ has independent increments on each interval $[i,i+1]$ (and the dimension $d=3$ has also become irrelevant). Moreover one may expect the additional quadratic weighting in \eqref{eq:confinement-example} to yield smaller variance paths under $\wt\bbP_{[0,T]}^{(\eta)}$ than for the base measure $\bbP_{[0,T]}^{(\eta)}$. In fact as shown in Section~\ref{sec:first-step-main} \revedit{(see Lemma~\ref{lem:unif-good}), for an absolute constant $C$ we have} 
\begin{equation}
\label{eq:example-variance-bound}
    \bbE^{\wt\bbP_{[i,i+1]}^{(\eta)}}[\sup_{t\in [i,i+1]}\|\bB_{t}-\bB_i\|^2]
    \leq
    \revedit{\log(\alpha T)}\sqrt{\frac{CR^3}{\alpha}}.
\end{equation}
\revedit{Recalling \eqref{eq:confinement-basic-example}, we have the simple bound
\begin{align*}
    \bbE^{\wh\bbP_{\alpha,T}^{(A,\eta)}}[\|\bB_T\|^2 \cdot (1-\one_{K_{R,[0,T]}^{(\eta)}})]
    &\leq
    \bbE^{\wh\bbP_{\alpha,T}^{(A,\eta)}}[\|\bB_T\|^4]^{1/2}
    \cdot 
    \lt(1-
    \wh\bbP_{\alpha,T}^{(A,\eta)}[K_{R,[0,T]}^{(\eta)}]
    \rt)^{1/2}
    \\
    &\leq
    \bbE^{\bbP_{[0,T]}^{(\eta)}}[\|\bB_T\|^4]^{1/2}
    \cdot 
    \alpha^{-5}
    \\
    &\leq O(T\alpha^{-5}).
\end{align*}
Combining with the Gaussian correlation inequality and the fact that $\wt\bbP$ has i.i.d. increments,  
\begin{equation}
\label{eq:example-key-bound}
\begin{aligned}
    \bbE^{\wh\bbP_{\alpha,T}^{(A,\eta)}}[\|\bB_{T}\|^2]
    &\leq 
    \bbE^{\wh\bbP_{\alpha,T,R}^{(A,\eta)}}[\|\bB_{T}\|^2]
    +
    \bbE^{\wh\bbP_{\alpha,T}^{(A,\eta)}}[\|\bB_T\|^2 \cdot (1-\one_{K_{R,[0,T]}^{(\eta)}})]
    \\
    &\stackrel{\eqref{eq:tv-bound-simple}}{\leq}
    \bbE^{\wh\bbP_{\alpha,T,R}^{(A,\eta)}}[\|\bB_{T}\|^2]
    + O(T\alpha^{-5})
    \\
    &\stackrel{\eqref{eq:comparison-example}}\leq \bbE^{\wt\bbP_{[0,T]}^{(\eta)}}[\|\bB_{T}\|^2] + O(T\alpha^{-5})
    \\
    &=
    T\cdot \bbE^{\wt\bbP_{[0,1]}^{(\eta)}}[\|\bB_{1}\|^2] 
    + 
    O(T\alpha^{-5})
    \\
    &\leq
    T\log(\alpha T)\cdot O\lt(
    \sqrt{\frac{R^3}{\alpha}}+\alpha^{-5}
    \rt)
    .
\end{aligned}
\end{equation}
}
If the logarithmic factors were not present, then \eqref{eq:example-key-bound} would already imply a lower bound of $\alpha^{1/2}$ for the effective mass. Moreover the exponent of $\alpha$ can be improved recursively. Namely it can be deduced from \eqref{eq:example-variance-bound} that
\begin{equation}
\label{eq:example-iterated-confinement}
    \revedit{\wt\bbP^{(\eta)}_{[0,T]}}[K^{(\eta)}_{R_2,[0,T]}]\geq 1-\alpha^{-10}
\end{equation}
holds for a smaller radius
\begin{equation}
\label{eq:R-iteration-example}
    R_2\leq O\lt((\log \alpha T)^{O(1)}\cdot \sqrt[4]{\frac{R^3}{\alpha}}\rt).
\end{equation}
\revedit{The idea here is that $\wt\bbP^{(\eta)}$ is a more confined Gaussian measure that the Brownian motion we started with.}
Then by \eqref{eq:comparison-example}, we find that \eqref{eq:example-iterated-confinement} holds also under $\wh\bbP_{\alpha,T,R}^{(A,\eta)}$, and by \eqref{eq:tv-bound-simple} a similar bound holds for the original measure $\wh\bbP_{\alpha,T}^{(A,\eta)}$.
This allows us to truncate further using the event $K_{R_2,[0,T]}^{(\eta)}$, then replace $\frac{\alpha}{2\text{e}R^3}$ in \eqref{eq:confinement-example} by $\frac{\alpha}{2\text{e} R_2^3}$ to obtain a yet more confined Gaussian measure, and so on.

Iterating eventually yields a bound of the form
\begin{equation}
\label{eq:example-variance-bound-iterated}
    \bbE[\|\bB_{i+1}-\bB_i\|^2]
    \leq
    \frac{\revedit{(\log \alpha T)^{O(1)}}}{\alpha^2},
\end{equation}
where the expectation is taken relative to a Gaussian measure which dominates ``most of'' $\wh\bbP_{\alpha,T}^{(A,\eta)}$, similarly to \eqref{eq:example-variance-bound}. Like $\wt\bbP_{[0,T]}^{(\eta)}$, this Gaussian measure has independent integer increments and so \eqref{eq:example-key-bound} extends to
\begin{equation}
\label{eq:basic-conclusion-example}
    \bbE^{\wh\bbP_{\alpha,T}^{(A,\eta)}}[\|\bB_T\|^2/T]\leq \frac{\revedit{(\log \alpha T)^{O(1)}}}{\alpha^2}.
\end{equation}

The estimate \eqref{eq:basic-conclusion-example} is unsatisfactory in two ways. The first is that due to the $T$-dependence of the logarithmic factors, it does not directly imply \emph{anything} about the effective mass due to the order in which the $T\to\infty$ and $\alpha\to\infty$ limits are taken.
\revedit{In fact while dependence on $\log\alpha$ seems unavoidable in the argument above, the only dependence on $\log T$ came from the union bound over $0\leq i\leq T-1$ in the first step \eqref{eq:confinement-basic-example}. The definition of $K_{R,[0,T]}^{(\eta)}$ can be refined to avoid this dependence.} Namely instead of treating all of $[0,T]$ at once, we split off the rare ``bad" intervals $[i,i+1]$ with large fluctuations and apply the main argument above only to ``good'' intervals, with a small fraction of intervals becoming bad at each iteration step. 
\revedit{It is not obvious that such an argument is possible: in using the Gaussian correlation inequality one does have to reason about the entire path on $[0,T]$ at once.}
The first stage is carried out in Section~\ref{sec:decompose} by decomposing the law $\bbP^{(\eta)}_{[0,1]}$ of $\bB_{[0,1]}$ into a mixture
\begin{equation}
\label{eq:mixture-informal}
    \bbP^{(\eta)}_{[0,1]}
    =
    (1-\delta)\nu_{\good_1}+\delta\nu_{\bad_1}
\end{equation}
satisfying $\delta\leq\alpha^{-10}$ and
\begin{equation}
\label{eq:GCI-dom-main-example}
    \nu_{\good_1}\preceq \bbP^{(\eta)}_{[0,1]},\quad\nu_{\bad_1}\preceq (\bbP^{(\eta)}_{[0,1]})^{\times 2}.
\end{equation}
Crucially $\nu_{\good_1}$ is supported inside
\revedit{an analogously defined $K^{(\eta)}_{R_1,i}\subseteq C^{(\eta)}([i,i+1];\bbR^3)$ for $R_1\leq O(\sqrt{\log \alpha})$ now independent of $T$, while $\nu_{\bad_1}$ has only degraded by a dilation factor of $2$}. Taking the $T$-th power of \eqref{eq:mixture-informal} then yields a decomposition of $\bbP^{(\eta)}_{[0,T]}$ into $2^T$ product measures, each of which is dominated by a corresponding product of the factors in \eqref{eq:GCI-dom-main-example}. The iteration (analogous to \eqref{eq:R-iteration-example} and below) is implemented in Section~\ref{sec:iterate} by a recursive generalization of \eqref{eq:mixture-informal} and only pays $\log(\alpha)$ factors at each stage.

The second issue with \eqref{eq:basic-conclusion-example} is that the dependence on $\alpha$ is quadratic rather than quartic. To improve the exponent we need to look beyond the single-time fluctuations we have considered so far. As shown in \cite{mukherjee2020identification}, even on short $O(\alpha^{-2})$ time-scales, $\bB_t$ behaves as a stationary \emph{Pekar process} with fluctuations of order $\alpha^{-1}$ under the path measure $\wh\bbP_{\alpha,T}$. However the centers of these Pekar processes are expected to vary much more slowly, see e.g. \cite[Section 4]{spohn1987effective}. This suggests that the increments $\|\bB_{i+1}-\bB_i\|^2$ should be mostly ``noise'' from short-time fluctuations, and so
\begin{equation}
\label{eq:better-improvement}
    \bbE^{\wh\bbP_{\alpha,T}^{(A,\eta)}}[\|\bB_T\|^2/T]
    \ll
    \bbE^{\wh\bbP_{\alpha,T}^{(A,\eta)}}[\|\bB_{i+1}-\bB_i\|^2].
\end{equation}
Our final argument in Section~\ref{sec:slow-oscillation} follows this intuition, showing a version of \eqref{eq:better-improvement} for the dominating Gaussian measures. In particular, for an adjacent pair of ``good'' intervals $[i,i+1],~[i+1,i+2]$ we consider the ``smoothed fluctuations'' 
\begin{equation}
\label{eq:smoothed-fluctuation-example}
    \int_{i+1}^{i+2} \bB_t~\de t-\int_{i}^{i+1} \bB_t~\de t.
\end{equation}
We show these smoothed fluctuations indeed obey improved upper bounds, which allows us to prove Theorem~\ref{thm:main}. Here and only here, it is important for the relevant Gaussian measures to also include interactions between adjacent intervals $[i,i+1],~[i+1,i+2]$.

At a high level, the local fluctuation estimates on $\|\bB_{i+1}-\bB_t\|^2$ appear as a priori bounds in this final argument. However we emphasize that the Gaussian correlation inequality permeates our whole proof. In particular the intermediate steps really need to be statements of Gaussian domination; fluctuation bounds for the polaron path measure proved in a different way would not suffice. As an illustration of the subtlety, we do not know how to deduce any effective mass lower bound directly from an upper bound on the $T=1$ variance
\[
\bbE^{\wh\bbP_{\alpha,1}}[\|\bB_1\|^2].
\]
This is because while the dominating Gaussian measures have independent increments by construction, the polaron path measure itself could have highly correlated increments across time. 

\subsection{Notations for Path Measures and More}
\label{subsec:notation}

We will consider a large number of Gaussian and non-Gaussian measures on continuous paths. The conventions that we have aimed to follow are summarized below (though the main arguments are intended to be unambiguous on their own). We first point out that since $\bbP$ always denotes a variant of Wiener measure, we \emph{never} write $\bbP^{\nu}$ to denote probability taken relative to a probability measure $\nu$. Instead, such superscripts or subscripts on $\bbP$ always denote different probability measures, usually obtained by reweighting, as discussed more below. However we \emph{do} write $\bbE^{\nu}$ for expectations relative to $\nu$. Note that all path measures we consider in subsequent sections are discrete-time and supported on $C^{(\eta)}([a,b];\bbR^3)$, except in Subsection~\ref{subsec:kernel}.

Recall that $\bbP_{[a,b]}$ denotes the law of $3$-dimensional Brownian motion $\bB_t$ restricted to times $t\in [a,b]$, while $\bbP^{(\eta)}_{[a,b]}$ refers to the corresponding law on piecewise-linear processes, still denoted $\bB_t$. 
\revedit{It will often be useful to consider path measures such as $\bbP^{(\eta)}$ to be product measures, using independence of increments to write expressions of the form}
\[
    \bbP^{(\eta)}_{[0,T]}=\prod_{i=0}^{T-1}\bbP^{(\eta)}_{[i,i+1]}.
\]
\revedit{This is a slight abuse of notation since a path sampled from $\bbP^{(\eta)}_{[0,T]}$ is of course not drawn from a product measure when viewed as an element of $C^{(\eta)}([0,T];\bbR^3)$. However it makes formal sense if we implicitly consider} $\bbP^{(\eta)}_{[a,b]}$ to be a stochastic process ``modulo global shift'' and identify $\bB_t$ with its increments $\bB_{(j+1)\eta}-\bB_{j\eta}$. 
\revedit{In this way, a full path is obtained from its factors by joining individual short paths on each $[i,i+1]$.} However we \revedit{will always} fix $\bB_0=(0,0,0)$, so this \revedit{point of view is unnecessary} if $a=0$. 

Hat notations $\wh\bbP^{(\eta)}$ always indicate the presence of the polaron interaction factor. Tilde notations of the form $\wt\bbP^{(\eta)}_{[i,i+1]}$ indicate Gaussian measures analogous to $\wt\bbP_{[0,T]}^{(\eta)}$ in the previous subsection.

Measures denoted in bold font such as $\bP^{(\eta)}$ or $\bP^{(A,\eta)}$ are always components of a mixture distribution over $C^{(\eta)}([0,T];\bbR^3)$, obtained by taking $T$-th powers of \eqref{eq:mixture-informal} or similar. These are always indexed by a sequence $\gamma$ of length $T$, so there are many distinct measures $\bP^{(\eta)}_{\gamma}$. Correspondingly, $\wh\bP^{(A,\eta)}_{\gamma}$ denotes the reweighting of such product measures by the polaron interaction factor, while $\wt\bP^{(\eta)}_{\gamma}$ is always a Gaussian measure and typically dominates the corresponding $\wh\bP^{(A,\eta)}_{\gamma}$. All our main arguments work with cut-off, discretized versions of the polaron interaction and apply for any $A\geq\alpha$ and $\eta$ sufficiently small.

These measures will often be reweighted in other ways using the notation of Definition~\ref{def:reweight}. In such cases, the restriction to $t\in [i,i+1]$ always ``comes first''. For example, \eqref{eq:reweight} defines the measure $\wt\bbP^{(\eta)}_{i,\beta}$ as a reweighting of $\bbP^{(\eta)}_{[i,i+1]}$. One could imagine instead reweighting $\bbP^{(\eta)}$ by a factor depending on values of $\bB_t$ for $t\notin [i,i+1]$ and then restricting to $t\in [i,i+1]$, thereby obtaining a different measure. However none of the measures we assign symbols to are constructed using this order of operations. The relative order of operations for reweighting and dilation $(\cdot)^{\times 2}$ can go either way, and will always be made explicit using parentheses.

Finally, we will often reweight a decomposition of a measure by a weight function. In doing so we simply mean that if
\begin{equation}
\label{eq:reweight1}
    \nu=\sum_{j=1}^k p_j\nu_j
\end{equation}
for $p_j\geq 0$ and $\sum_{j=1}^k p_j=1$ is a mixture representation of the probability measure $\nu$, and if $f\in L^1(\nu)$ is non-negative with non-zero expectation, then 
\begin{equation}
\label{eq:reweight2}
    \nu^{(f)}=\sum_{j=1}^k q_j \nu_j^{(f)}
\end{equation}
holds for the new weights
\begin{equation}
\label{eq:reweight3}
    q_j=p_j\cdot\frac{\bbE^{\nu_j}[f]}{\bbE^{\nu}[f]}.
\end{equation}
We will refer to \eqref{eq:reweight2} as a reweighting of \eqref{eq:reweight1}. It may be intuitively helpful to view the formula \eqref{eq:reweight3} as a version of Bayes' rule.

Note that if $(\nu,f)$ or equivalently $(\nu,\nu^{(f)})$ are given, then \eqref{eq:reweight1} uniquely determines \eqref{eq:reweight2}. In particular reweighting by a product of functions $f_1f_2\dots f_k$ can be done in any number of steps with the same result.
Furthermore, reweighting commutes with refining or coarsening a decomposition as stated below; the proof is omitted.

\begin{proposition}
\label{prop:refine-partition}
    Let $S_1,\dots,S_{\ell}$ be a partition of $[k]$. 
    Suppose \eqref{eq:reweight1} holds and let
    \[
    \nu_{S_i}=p_{S_i}^{-1}\sum_{j\in S_i} p_j \nu_j
    \]
    where $p_{S_i}=\sum_{j\in S_i}p_j$.
    Then for $q_{S_i}=\sum_{j\in S_i} q_j$ with $q_j$ as in \eqref{eq:reweight3},
    \[
    \nu
    =
    \sum_{i=1}^{\ell}
    p_{S_i} \nu_{S_i},
    \quad\quad\text{and}
    \quad\quad
    \nu^{(f)}
    =
    \sum_{i=1}^{\ell}
    q_{S_i}\nu_{S_i}^{(f)}.
    \]
\end{proposition}

\section{Mixture Decomposition of the Polaron Path Measure}
\label{sec:decompose}

\revedit{
The next lemma is crucial for us. For any symmetric convex body $K\subseteq X$ with high probability relative to a Gausian measure $\mu$, it gives a decomposition of $\mu$ into a \textbf{good} component supported inside a constant dilation of $K$ and dominated by $\mu$, and a \textbf{bad} component which is still dominated by a factor-two dilation of $\mu$. 
}

\begin{lemma}
\label{lem:good-decomposition}
    \revedit{
    There exists an absolute constant $C_{\ref{lem:good-decomposition}}\geq 100$ such that the following holds.}
    Let $\mu$ be a centered Gaussian measure on \revedit{the finite-dimensional real vector space} $X$, and $K\subseteq X$ a symmetric convex set with $\mu(K)\geq 1-\delta$ for some $\delta\leq 0.1$. 
    There exists a decomposition 
    \[
    \mu=(1-\delta')\nu_{\good_1}+\delta'\nu_{\bad_1}
    \]
    of $\mu$ into a mixture of probability measures $\nu_{\good_1},\nu_{\bad_1}$
    such that:
    \begin{enumerate}[label=(\roman*)]
        \item 
        \label{it:small-amount-bad}
        $\delta'\leq \delta$. 
        \item
        \label{it:good-support}
        $\supp(\nu_{\good_1})\subseteq C_{\ref{lem:good-decomposition}} K$.
        \item 
        \label{it:good-domination}
        $\nu_{\good_1}\preceq \mu$
        \item 
        \label{it:bad-domination}
        $\nu_{\bad_1}\preceq \mu^{\times 2}$.
    \end{enumerate}
\end{lemma}

\begin{proof}
\revedit{By a coordinate change, we may assume $X$ is an inner product space such that $\mu$ is a standard Gaussian with identity covariance $I_X$. We write $d(\cdot,\cdot)$ below for the associated Euclidean distance, and abbrevate $d(x)=d(x,K)$ for the distance to $K$. Since $K$ is a convex body with positive measure, we can and do assume it is closed so that a closest point to any $x\in X$ exists. Finally we let $R=C_1\sqrt{\log(1/\delta)}$ for a large absolute constant $C_1$, and define 
\[
\bbB_R(S)=\{x\in X~:~d(x,S)\leq R\},
\quad
S\subseteq X.
\]
We will choose $C_{\ref{lem:good-decomposition}}$ below depending on $C_1$. Note that the condition $C_{\ref{lem:good-decomposition}}\geq 100$ is without loss of generality.
}

\revedit{The desired decomposition can be constructed via}
\begin{equation}
\label{eq:nu-good-bad-construction}
\begin{aligned}
    \de\nu_{\bad_1}(x)&\propto e^{-\sigma(d(x))}\de\mu(x),
    \\
    \de\nu_{\good_1}(x)&\propto \big(1-e^{-\sigma(d(x))}\big)\de\mu(x)
\end{aligned}
\end{equation}
for any $C^2_b$ function $\sigma:[0,\infty]\to [0,R^2]$ satisfying:
\begin{enumerate}[label=(\alph*)]
    \item 
    \label{it:sigma-decrease}
    $\sigma$ is non-increasing;
    \item 
    \label{it:sigma-small}
    $\sigma(r)=R^2$ for $r\in [0,1]$;
    \item
    \label{it:sigma-compact}
    $\sigma(r)=0$ for $r\geq 3C_1 R$;
    \item 
    \label{it:sigma-lipschitz}
    $|\sigma'(r)|\leq (r-1)/C_1$ for all $r\geq 1$;
    \item 
    \label{it:sigma-smooth}
    $|\sigma''(r)|\leq 1/C_1$ for all $r\geq 0$.
\end{enumerate}
An explicit such $\sigma$ can be constructed by convolving a $C^{\infty}_c$ bump function with the following joining $\wt\sigma$ of two quadratics:
\[
    \wt\sigma(r)=
    \begin{cases}
    R^2,\quad\quad\quad\quad\quad 0\leq r\leq 1
    \\
    R^2-\frac{(r-1)^2}{2C_1^2},\quad 1\leq r\leq C_1R+1
    \\
    \frac{(2C_1R-r+1)^2}{2C_1^2},\quad C_1R+1\leq r\leq 2C_1 R+1
    \\
    0,\quad\quad\quad\quad\quad\quad r\geq 2C_1R+1
    \,.
    \end{cases}
\]
We claim that for such $\sigma$, the resulting $\nu_{\good_1},\nu_{\bad_1}$ in \eqref{eq:nu-good-bad-construction} satisfy the conclusions of Lemma~\ref{lem:good-decomposition}. First, \ref{it:small-amount-bad} holds since by choice of $R$ the weight of $\nu_{\bad_1}$ in $\mu$ is
\[
    \delta'=\bbE^{x\sim \mu}\lt[e^{-\sigma(d(x))}\rt]
    \leq
    e^{-R^2}+1-\mu[\bbB_R(K)]
    \leq \delta.
\]
\revedit{
The latter bound holds uniformly on $\delta\in [0,0.1]$ since $R=C_1\sqrt{\log(1/\delta)}$ for $C_1$ a large absolute constant. Namely we have by definition $e^{-R^2}\leq \delta/2$, while the Gaussian isoperimetric inequality and $\mu(K)\geq 1/2$ implies 
\[
1-\mu[\bbB_R(K)]
\geq 
\frac{1}{\sqrt{2\pi}}
\int_R^{\infty} e^{-u^2/2}~\de u \leq \delta/2
\]
for large $C_1$.
}
Also $\sigma$ is decreasing, so the Radon--Nikodym derivative 
\[
    \frac{\de \nu_{\good_1}}{\de \mu}\propto 1-e^{-\sigma(d(x))}
\]
is origin-symmetric and \revedit{quasi-concave}. Hence conclusion~\ref{it:good-domination} holds: $\nu_{\good_1}\preceq\mu$. 

\revedit{
Since $\mu(K)=1-\delta\geq 0.9$ and $K$ is symmetric convex, we next deduce the in-radius lower bound (where $0$ denotes the origin in $X$):
\begin{equation}
\label{eq:inradius-LB}
    \bbB_{R/C_1^2}(0)\subseteq K.
\end{equation}
Indeed if $x\in X\backslash K$, the Hahn--Banach theorem guarantees existence of a hyperplane $H\subseteq X$ containing $x$ such that $K$ lies entirely on one side of $H$. 
If for sake of contradiction we had $\|x\|\leq R/C_1^2$, then we would deduce $d(0,H)\leq d(0,x)\leq R/C_1^2$. Since $K$ is contained on one side of $H$, its $\mu$-measure can be hence bounded by the Gaussian integral
\[
    \mu(K) \leq 
    \frac{1}{\sqrt{2\pi}}
    \int_{-\infty}^{R/C_1^2}
    e^{-u^2/2}~\de u.
\]
Since $R/C_1^2 = C_1^{-1}\sqrt{\log(1/\delta)}$, for $C_1$ a large absolute constant, the latter integral is smaller than $1-\delta$, uniformly in $\delta\in [0,0.1]$. (For instance, $C_1=10$ clearly suffices for small enough $\delta$, and taking $C_1$ large enough depending on $\delta_0>0$ clearly suffices uniformly on $\delta\in [\delta_0,0.1]$.)

Together with \eqref{eq:inradius-LB}, property~\ref{it:sigma-compact} of $\sigma$ implies that $\supp(\nu_{\good_1})\subseteq 4C_1^3K$. Thus conclusion \ref{it:good-support} holds for some $C_{\ref{lem:good-decomposition}}$.

It remains to verify $\nu_{\bad_1}\preceq\revedit{\mu}^{\times 2}$. We show the stronger statement that the Radon--Nikodym derivative $\frac{\de \nu_{\bad_1}}{\de \revedit{\mu}^{\times 2}}$ is log-concave. Since
\begin{align*}
    \frac{\de \nu_{\bad_1}}{\de \revedit{\mu}^{\times 2}}
    &=
\frac{\de \nu_{\bad_1}}{\de\mu}\cdot\frac{\de\mu}{\de \revedit{\mu}^{\times 2}}
    \\
    &\stackrel{\eqref{eq:nu-good-bad-construction}}{\propto}
    \exp\lt(-\sigma(d(x))-\frac{\|x\|^2}{2}+\frac{\|x\|^2}{8}\rt)
    \\
    &=
    \exp\lt(-\sigma(d(x))-\frac{3\|x\|^2}{8}\rt)
\end{align*}
this amounts to the convexity of $x\mapsto \sigma(d(x))+\frac{3x^2}{8}$.

To verify this convexity, we let $p\in (0,1)$ and take $x,z\in X$ distinct, and set $y=px+(1-p)z$. 
Since we assumed previously that $K$ is closed, let $k_y$ be the closest point in $K$ to $y$. Then since $\sigma$ is decreasing,
\begin{align*}
    \sigma(d(x)) 
    \geq 
    \sigma(d(x,k_y))
\end{align*}
and similarly for $z$. Hence it will certainly suffice to show that 
\begin{equation}
\label{eq:suffices}
\begin{aligned}
    p\sigma(d(x,k_y)) + (1-p)\sigma(d(z,k_y))-\sigma(d(y,k_y))
    &\stackrel{?}{\geq}
    -\frac{3}{8}
    \big(p\|x\|^2+(1-p)\|z\|^2 - \|y\|^2\big)
    \\
    &=
    -\frac{3p(1-p)\|x-z\|^2}{8}.
\end{aligned}
\end{equation}
Next let $a$ be the distance from $k_y$ to the (bi-infinite) line $\overline{xz}$, and $o$ the closest point to $k_y$ on this line. Let $w_x,w_y,w_z$ be the signed distances from $x,y,z$ to $o$ with some arbitrary but consistent choice of sign, so that $w_y=pw_x+(1-p)w_z$ and $|w_x-w_z|=\|x-z\|$. Then we have 
\begin{align*}
    d(x,k_y)^2 &= a^2 + w_x^2,
    \\
    d(y,k_y)^2 &= a^2 + (pw_x+(1-p)w_z)^2,
    \\
    d(z,k_y)^2 &= a^2 + w_z^2.
\end{align*}
Hence fixing $a>0$ (the case $a=0$ can be obtained as a limit or treated separately), we are led to define the function $f(w)=\sigma(\sqrt{a^2+w^2})$. One readily computes:
\begin{align*}
    f'(w)
    &=
    \frac{w}{\sqrt{a^2+w^2}}\cdot \sigma'(\sqrt{a^2+w^2}),
    \\
    f''(w)
    &=
    \frac{w^2}{a^2+w^2}\cdot \sigma''(\sqrt{a^2+w^2})
    +
    \frac{a^2}{(a^2+w^2)^{3/2}}\cdot\sigma'(\sqrt{a^2+w^2}).
\end{align*}
In particular, properties \ref{it:sigma-lipschitz}, \ref{it:sigma-smooth} of $\sigma$ imply that $\sup_{w\geq 0}|f''(w)|\leq 2/C_1$.
Then by Taylor's theorem the left-hand side of \eqref{eq:suffices} is 
\begin{align*}
    pf(w_x)+(1-p)f(w_z)-f(pw_x+(1-p)w_z)
    &\geq 
    -\frac{p(1-p)}{2}\cdot (w_x-w_z)^2 \sup_{w\geq 0}|f''(w)|
    \\
    &\geq 
    -\frac{p(1-p)\|x-z\|^2}{C_1}
\end{align*}
Comparing with the right-hand side of \eqref{eq:suffices} completes the proof for suitable $C_1$.
}
\end{proof}

We now apply Lemma~\ref{lem:good-decomposition} to each $\bbP^{(\eta)}_{[i,i+1]}$ as follows. 
\revedit{For $C_{\ref{lem:good-decomposition}}$ as in Lemma~\ref{lem:good-decomposition} and $i\in [T]= \{0,1,\dots,T-1\}$, define:
\begin{align}
\nonumber
    K^{(\eta)}_{R,i}&=\lt\{\bB\in C^{(\eta)}([0,T];\bbR^3)~:~\sup_{t,s\in [i,i+1]} \|\bB_t-\bB_s\|\leq R\rt\},
    \\
\nonumber
    K^{(\eta)}_{R}&=K^{(\eta)}_{R,0}
    \\
\label{eq:R1}
    R_1&=C_{\ref{lem:good-decomposition}}^2\sqrt{\log(\alpha)},
    \\
\nonumber
    \delta_1&=1-\bbP^{(\eta)}_{[i,i+1]}[\revedit{K^{(\eta)}_{R_1/C_{\ref{lem:good-decomposition}},i}}].
\end{align}
Then taking $(K,\delta)=(\revedit{K^{(\eta)}_{R_1/C_{\ref{lem:good-decomposition}},i}},\delta_1)$ in Lemma~\ref{lem:good-decomposition} yields for each $i\in [T]$ a decomposition
\[
    \bbP^{(\eta)}_{[i,i+1]}=(1-\delta_1')\nu_{i,\good_1}+\delta_1'\nu_{i,\bad_1}
\]
with (using $C_{\ref{lem:good-decomposition}}\geq 100$ to obtain the first line):
}
\begin{equation}
\label{eq:path-decomp-v1-result}
\begin{aligned}
    \delta_1'\leq \delta_1&\leq \alpha^{-10},
    \\
    \supp(\nu_{i,\good_1})&\subseteq K^{(\eta)}_{R_1,i},
    \\
    \nu_{i,\good_1}&\preceq \bbP^{(\eta)}_{[i,i+1]},
    \\
    \nu_{i,\bad_1}&\preceq (\bbP^{(\eta)}_{[i,i+1]})^{\times 2}
    .
\end{aligned}
\end{equation}

Of course, we may assume these decompositions are identical up to indexing as $i$ varies. Since $\bbP^{(\eta)}_{[0,T]}$ has independent increments, we obtain a corresponding product decomposition
\begin{align}
\label{eq:product-decomp-basic}
    \bbP^{(\eta)}_{[0,T]}
    &=
    \sum_{\gamma\in \{\good_1,\bad_1\}^T}
    w_{R_1}(\gamma)
    \bP^{(\eta)}_{\gamma};
    \\
\nonumber
    \bP^{(\eta)}_{\gamma}&\equiv\prod_{i=0}^{T-1} \nu_{i,\gamma_i},
    \\
\nonumber
    w_{R_1}(\gamma)&=(1-\delta_1')^{|\gamma_{\good_1}|}(\delta_1')^{|\gamma_{\bad_1}|}
    .
\end{align}
Here we have used the notation
\begin{equation}
\label{eq:gamma-good-1-def}
\gamma_{\good_1}\equiv\{i\in [T]~:~\gamma_{i}=\good_1\},
\quad\quad\quad
\gamma_{\bad_1}\equiv\{i\in [T]~:~\gamma_{i}=\bad_1\}
\end{equation}
Next recalling \eqref{eq:W-def}, let
\[
    \de\wh\bP^{(A,\eta)}_{\gamma}(\bB)
    \propto
    W_{\alpha,T}^{(A)}(\bB)
    ~\de\bP^{(\eta)}_{\gamma}(\bB)
\]
be the reweighting of $\bP^{(\eta)}_{\gamma}$ by the cut-off interaction factor $ W_{\alpha,T}^{(A)}$. Then \eqref{eq:product-decomp-basic} becomes a decomposition of the polaron path measure:
\revedit{
\begin{equation}
\label{eq:product-decomp-polaron}
\begin{aligned}
    \wh\bbP_{\alpha,T}^{(A,\eta)}
    &=
    \sum_{\gamma\in \{\good_1,\bad_1\}^T}
    \wh w_{R_1}^{(A,\eta)}(\gamma)
    \wh\bP^{(A,\eta)}_{\gamma};
    \\
    \wh w_{R_1}^{(A,\eta)}(\gamma)
    &=
    w_{R_1}(\gamma)\cdot 
    \frac{\bbE^{\bP^{(\eta)}_{\gamma}}\lt[W_{\alpha,T}^{(A)}\rt]}{\bbE^{\bbP^{(\eta)}_{[0,T]}}\lt[W_{\alpha,T}^{(A)}\rt]}.
\end{aligned}
\end{equation}
}
Next we show the label $\good_1$ still predominates after reweighting.

\begin{lemma}
\label{lem:mostly-good}
    The weights $\wh w_{R_1}^{(A,\eta)}(\cdot)$ just constructed satisfy
    \[
    \sum_{\gamma\in \{\good_1,\bad_1\}^T}
    \wh w_{R_1}^{(A,\eta)}(\gamma) |\gamma_{\good_1}|
    \geq
    T(1-\alpha^{-10}).
    \]
\end{lemma} 

\begin{proof}
    Fix $i\in [T]$ and let 
    \begin{equation}
    \label{eq:nu-extend}
    \nu_{i,\good_1,[0,T]}=
    \bbP^{(\eta)}_{[0,i]}\times {\nu}_{i,\good_1}\times \bbP^{(\eta)}_{[i+1,T]}.
    \end{equation}
    (Here the product indicates that increments are independently joined as mentioned in Subsection~\ref{subsec:notation}.)
    From \eqref{eq:path-decomp-v1-result}, we have
    \[
    \nu_{i,\good_1,[0,T]}\preceq \bbP^{(\eta)}_{[0,T]}.
    \]
    Moreover $W_{\alpha,T}^{(A)}$ is symmetric and quasi-concave, so 
    \begin{align*}
    \sum_{\gamma\in \{\good_1,\bad_1\}^T:~\gamma_i=\good_1}
    \wh w_{R_1}^{(A,\eta)}(\gamma)
    &=
    (1-\delta_1')
    \cdot
    \frac{\int W_{\alpha,T}^{(A)}
    ~\de \nu_{i,\good_1,[0,T]}}
    {
    \int W_{\alpha,T}^{(A)}
    ~\de \bbP^{(\eta)}_{[0,T]}
    }
    \\
    &\geq
    (1-\delta_1')
    \geq 
    1-\alpha^{-10}
    .
    \end{align*}
    Summing over $0\leq i\leq T-1$ completes the proof.
\end{proof}

\section{A One-Step Estimate}
\label{sec:first-step-main}

We now define the Gaussian measures which will dominate the corresponding polaron component $\wh\bP^{(A,\eta)}_{\gamma}$. For each integer interval $[i,i+1]$, define the quadratic form
\begin{equation}
\label{eq:Qibeta}
    Q_i(\bB)\equiv
    \int_i^{i+1}\int_i^{i+1} \|\bB_t-\bB_s\|^2~\de t~\de s.
\end{equation}
For each $\beta\geq 0$, let $\wt\bbP^{(\eta)}_{i,\beta}=(\bbP^{(\eta)}_{[i,i+1]})^{\la\beta Q_i\ra}$, i.e.
\begin{equation}
\label{eq:reweight}
    \de\wt\bbP^{(\eta)}_{i,\beta}(\bB)
    \propto
    \exp\lt(
    -\beta \int_i^{i+1} \int_i^{i+1}
    \|\bB_{s}-\bB_{t}\|^2 
    ~\de t~\de s
    \rt)
    \de\bbP^{(\eta)}_{[i,i+1]}( \bB).
\end{equation}
\revedit{
We set 
\[
\beta_1=\frac{\alpha}{16\text{e}^2R_1^3},
\]
where $R_1=C_{\ref{lem:good-decomposition}}^2\sqrt{\log(\alpha)}$ as in \eqref{eq:R1}.
(The factor $16\text{e}^2$ rather than $2\text{e}$ will be convenient later in Lemma~\ref{lem:gaussian-domination-k-step}.)
} 
Then define
\[
\wt\bbP^{(\eta)}_{i,\good_1}=\wt\bbP^{(\eta)}_{i,\beta_1}.
\]
The corresponding $\bad_1$ distribution is simply a dilation of the base measure: 
\[
    \wt\bbP^{(\eta)}_{i,\bad_1}\equiv (\bbP^{(\eta)}_{[i,i+1]})^{\times 2}.
\]
For each $\gamma\in \{\good_1,\bad_1\}^T$ we consider the path measure
\begin{equation}
\label{eq:dominating-product}
    \wt\bP^{(\eta)}_{\gamma}
    =\prod_{i=0}^{T-1}
    \wt\bbP^{(\eta)}_{i,\gamma_i}.
\end{equation}

\begin{lemma}
\label{lem:gaussian-domination-one-step}
    For $A\geq\alpha$ and any 
    $\gamma\in \{\good_1,\bad_1\}^T$, we have
    \[
    \wh\bP^{(A,\eta)}_{\gamma}\preceq \wt\bP^{(\eta)}_{\gamma}.
    \]
    In particular,
    \[
    \bbE^{ \wh\bP^{(A,\eta)}_{\gamma}}[\|\bB_T\|^2]
    \leq
    \bbE^{ \wt\bP^{(\eta)}_{\gamma}}[\|\bB_T\|^2].
    \]
\end{lemma}

\begin{proof}
    We apply Proposition~\ref{prop:gaussian-domination-general}. Specifically, with $\bbQ^{(\eta)}_{i,\good_1}=\bbP^{(\eta)}_{[i,i+1]}$ and $\bbQ^{(\eta)}_{i,\bad_1}=(\bbP^{(\eta)}_{[i,i+1]})^{\times 2}$, we set:
    \begin{align*}
    \bbQ
    &=
    \bQ^{(\eta)}_{\gamma}
    \equiv
    \prod_{i\in [T]}\bbQ^{(\eta)}_{i,\gamma_i},
    \\
    \bbQ^{\dagger}&=\bP_{\gamma}^{(\eta)},
    \\
    \wt\bbQ
    &=
    \wt\bP_{\gamma}^{(\eta)},
    \\
    \wh\bbQ
    &=
    \revedit{\wh\bP_{\gamma}^{(A,\eta)}}.
    \end{align*}
    Parts~\ref{it:good-domination} and \ref{it:bad-domination} of Lemma~\ref{lem:good-decomposition} imply that $\nu_{i,\gamma_i}\preceq \bbQ^{(\eta)}_{i,\gamma_i}$ for each $i\in [T]$. Taking a product over $i$ then yields $\bbQ^{\dagger}\preceq\bbQ$. The Radon--Nikodym derivative $\de\wh\bbQ/\de\bbQ^{\dagger}$ is proportional to $W^{(A)}_{\alpha,T}$ as required, while $\de\wt\bbQ/\de\bbQ$ takes the required form with $F(s,t)=\beta_1$ if $s,t$ are in the same interval $[i,i+1]$ and $\gamma_i=\good_1$, and $F(s,t)=0$ otherwise. The condition \eqref{eq:F-bound} then holds by Lemma~\ref{lem:good-decomposition}, part~\ref{it:good-support} and the definition of $\beta_1$. Thus Proposition~\ref{prop:gaussian-domination-general} applies, \revedit{with the function $R(s,t)=R_1$ if $s,t\in [i,i+1]$ for some $i$, and $R(s,t)=+\infty$ otherwise}. This completes the proof.
\end{proof}

Our next goal is to bound the variance $\bbE^{\wt\bbP^{(\eta)}_{0,\beta}}[\|\bB_1\|^2]$ on good intervals. This is done in the following important estimate. 

\begin{lemma}
\label{lem:gain}
    \revedit{There exists an absolute constant $C_{\ref{lem:gain}}$ such that the following holds.
    For any $\beta\geq 2$, if $\eta>0$ is sufficiently small,}
    \[
    \sup_{t\in [0,1]}
    \bbE^{\wt\bbP^{(\eta)}_{0,\beta}}[\|\bB_t\|^2]
    \leq C_{\ref{lem:gain}}/\sqrt{\beta}.
    \]
\end{lemma}

We pause to record the suboptimal bound $m_{\eff}(\alpha)\geq \frac{c\alpha^{1/2}}{(\log\alpha)^{3/4}}$ using what we have seen so far. 

\begin{corollary}
\label{cor:suboptimal-bound}
    \revedit{There exists an absolute constant $C$ such that the following holds.} For $\alpha\geq  2$ and $A\geq\alpha$, and with $\eta>0$ sufficiently small,
    \[
    \lim_{T\to\infty}
    \bbE^{\wh\bbP^{(\eta)}_{\alpha,T}}
    \lt[\frac{\|\bB_T\|^2}{T}\rt]
    \leq
    \frac{C(\log\alpha)^{3/4}}{\alpha^{1/2}}
    .
    \]
\end{corollary}

\begin{proof}
     Using Lemma~\ref{lem:gain} and the independence and centeredness of increments $\wt\bP^{(\eta)}_{\gamma}$ on distinct intervals $[i,i+1]$, we obtain
    \begin{align*}
    \bbE^{\wh\bbP_{\alpha,T}^{(A,\eta)}}
    \lt[\frac{\|\bB_T\|^2}{T}\rt]
    &=
    \sum_{\gamma\in \{\good_1,\bad_1\}^T}
    \wh w_{R_1}^{(A,\eta)}(\gamma)
    \bbE^{\wh\bP^{(A,\eta)}_{\gamma}}
    \lt[\frac{\|\bB_T\|^2}{T}\rt]
    \\
    &\leq
    \sum_{\gamma\in \{\good_1,\bad_1\}^T}
    \wh w_{R_1}^{(A,\eta)}(\gamma)
    \bbE^{\wt\bP^{(\eta)}_{\gamma}}
    \lt[\frac{\|\bB_T\|^2}{T}\rt]
    \\
    &=
    \sum_{\gamma\in \{\good_1,\bad_1\}^T}
    \wh w_{R_1}^{(A,\eta)}(\gamma)
    \sum_{i=0}^{T-1}
    \bbE^{\wt\bbP^{(\eta)}_{i,\gamma_i}}
    \lt[\frac{\|\bB_i-\bB_{i+1}\|^2}{T}\rt]
    \\
    &\leq
    \sum_{\gamma\in \{\good_1,\bad_1\}^T}
    \wh w_{R_1}^{(A,\eta)}(\gamma)\cdot
    \revedit{
    \lt(C_{\ref{lem:gain}}\sqrt{\frac{16\text{e}^2 R_1^3}{\alpha}}+12|\gamma_{\bad_1}|\rt)
    }
    .
    \end{align*}
    \revedit{Here the last step follows by using Lemma~\ref{lem:gain} when $\gamma_i=\good_1$, and Lemma~\ref{lem:good-decomposition} part~\ref{it:bad-domination} when $\gamma_i=\bad_1$ (the value $12=2^2\cdot 3$ comes from the factor of two dilation and the fact that the processes live in $\bbR^3$).}
    Recalling Lemma~\ref{lem:mostly-good} and the definition~\eqref{eq:R1} of $R_1$ completes the proof.
\end{proof}

We will verify Lemma~\ref{lem:gain} in the next subsection by returning to continuous time and giving an exact series formula for $\bbE^{\wt\bbP_{0,\beta}}[\|\bB_t\|^2]$. However let us point out that is easy to guess Lemma~\ref{lem:gain}. Intuitively, the change of measure $\bbP^{(\eta)}\to \wt\bbP^{(\eta)}_{0,\beta}$ should be similar to reweighting Wiener measure $\bbP$ by $\exp\lt(-\beta \int_0^1 \|\bB_t\|^2 ~\de t\rt)$. It is easy to see that for ordinary Brownian motion,
\[
    \bbP\lt[\int_0^1 \|\bB_t\|^2~\de t\leq \eps\rt] = e^{-\Theta(1/\eps)}
\]
by considering the behavior of $\bB_t$ separately on each interval $[j\eps,(j+1)\eps]$. Therefore one expects that a $\wt\bbP^{(\eta)}_{0,\beta}$-typical path will satisfy $\int_0^1 \|\bB_t\|^2 ~\de t\asymp \beta^{-1/2}$ to minimize $\beta\eps+\frac{1}{\eps}$, which suggests the conclusion of Lemma~\ref{lem:gain}.

Finally we record the following uniform-in-time bound which is easily deduced from Lemma~\ref{lem:gain} and is important for the iterative argument in the next section. \revedit{We note that the failure probability below is a power of $\beta$ because $\alpha$ does not appear anywhere in the statement, but in all our applications it will be at most $\alpha^{-10}$ as in the rest of this paper.}

\begin{lemma}
\label{lem:unif-good}
    \revedit{
    There exists an absolute constant $C_{\ref{lem:unif-good}}$ such that the following holds. For any $\beta\geq 2$, for $\eta>0$ sufficiently small:
    }
    \begin{equation}
    \label{eq:unif-good}
    \wt\bbP^{(\eta)}_{0,\beta}
    \lt[\revedit{\sup_{s,t\in [0,1]}\|\bB_t-\bB_s\|} \geq \frac{
    \sqrt{\revedit{C_{\ref{lem:unif-good}}}\log(\beta)}}{\beta^{1/4}}\rt]
    \leq \beta^{-100}.
    \end{equation}
\end{lemma}

\begin{proof}
    Let $S_{\beta}=[0,1]\cap \lceil\beta^{-1}\rceil \bbZ$. Since $\bB_s$ is a centered Gaussian for each $s$, \revedit{union bounding over $S_{\beta}$ via} Lemma~\ref{lem:gain} yields
    \[
    \wt\bbP^{(\eta)}_{0,\beta}
    \lt[\sup_{s\in S_{\beta}}\|\bB_s\| \geq \frac{\revedit{100\sqrt{C_{\ref{lem:gain}}\log(\beta)}}}{\beta^{1/4}}\rt]
    \leq \beta^{-100}/2.
    \]
    Moreover Corollary~\ref{cor:quadratic-domination} implies that $\wt\bbP^{(\eta)}_{0,\beta}\preceq\bbP^{(\eta)}_{[0,1]}$ and so with $0\leq t_1,t_2\leq 1$ below,
    \begin{align*}
    &\wt\bbP^{(\eta)}_{0,\beta}
    \lt[\sup_{|t_1-t_2|\leq \beta^{-1}}\|\bB_{t_1}-\bB_{t_2}\| 
    \geq 
    \frac{\revedit{100}\sqrt{\log(\beta)}}{\beta^{1/4}}\rt]
    \\
    &\leq 
    \bbP^{(\eta)}_{[0,1]}
    \lt[\sup_{|t_1-t_2|\leq \beta^{-1}}\|\bB_{t_1}-\bB_{t_2}\| 
    \geq 
    \frac{\revedit{100}\sqrt{\log(\beta)}}{\beta^{1/4}}\rt]
    \\
    &\leq
    \beta^{-100}/2.
    \end{align*}
    \revedit{The last step follows by using the reflection principle to control the oscillation of Brownian motion on each interval $[k\beta^{-1},(k+2)\beta^{-1}]$, together with a union bound over integers $0\leq k\leq \beta$.}
    Combining completes the proof.
\end{proof}

\subsection{Proof of Lemma~\ref{lem:gain}}
\label{subsec:kernel}

We prove Lemma~\ref{lem:gain} by using a result of \cite{shepp1966radon} (see also \cite{cheridito2003representations}) to write $\bbE^{\wt\bbP_{0,\beta}}[\|\bB_t\|^2]$ as an explicit infinite series \revedit{given by the spectral expansion of a certain kernel}.\footnote{Only the scalar case is considered in \cite{shepp1966radon}, but the effect of reweighting by $Q_{\theta}$ is independent on the $d=3$ coordinates so this makes no difference.} By Proposition~\ref{prop:discrete-benign} these continuous-time estimates apply also for sufficiently small $\eta>0$. Recalling Definition~\ref{def:reweight}, we consider the law $\bbP_{[0,1]}^{\la Q_{\theta}\ra}=(\bbP_{[0,1]})^{\la Q_{\theta}\ra}$ for a quadratic form
\[
    Q_{\theta}(\bB)
    \equiv 
    \int_0^1 \int_0^1 \theta(t,s) ~\de \bB_t\de \bB_s.
\]
We will consider only cases in which $\theta\in L^2([0,1]^2)$ satisfies $\theta(s,t)=\theta(t,s)$ and $Q_{\theta}$ is always non-negative. To state the relevant result, we associate $\theta$ with the integral kernel operator $K:L^2([0,1])\to L^2([0,1])$ given by
\[
    K[\phi](x)=
    \int_0^1 \theta(x,y) \phi(y)~\de y.
\]
With $\bone$ the identity operator on $L^2([0,1])$, the resolvent kernel $\tilde \theta\in L^2([0,1]^2)$ is defined to have associated kernel operator $\wt K$ satisfying
\[
    \wt K=\bone-(\bone+K)^{-1}.
\]
Equivalently, $\theta$ and $\tilde \theta$ have orthonormal eigenfunction expansions related by:
\begin{equation}
\label{eq:resolvent-kernel}
\begin{aligned}
    \theta(t,s)&=\sum_{k\geq 1} \lambda_k v_k(t) v_k(s),
    \\
    \tilde \theta(t,s)&=\sum_{k\geq 1} \frac{\lambda_k}{1+\lambda_k} v_k(t) v_k(s).
\end{aligned}
\end{equation}

\begin{proposition}[{\cite{shepp1966radon}}]
\label{prop:shepp}
For $\theta$ as above, $\bbP_{[0,1]}^{\la Q_{\theta}\ra}$ is a centered Gaussian process $\bB_t=(B_t^{(1)},B_t^{(2)},B_t^{(3)})$ and the coordinate processes are i.i.d. with covariance 
\[
    \bbE^{\bbP_{[0,1]}^{\la Q_{\theta}\ra}}[B_t^{(1)}B_s^{(1)}]
    =
    \min(t,s)-
    \int_0^t \int_0^s
    \tilde \theta(u,v)~\de u~\de v
\]
for all $t,s\in [0,1]$.
\end{proposition}

\begin{proof} \revedit{Proposition~\ref{prop:discrete-benign} implies that uniformly on $t\in [0,1]$:}
\[
\lim_{\eta\to 0}
\bbE^{\wt\bbP^{(\eta)}_{0,\beta}}[\|\bB_t\|^2]
=
\bbE^{\wt\bbP_{0,\beta}}[\|\bB_t\|^2].
\]
Hence it suffices to work in continuous time and bound the right-hand side. We write $B_t=B_t^{(1)}$ throughout, and begin by expressing
\[
\int_0^1 \int_0^1 (B_t-B_s)^2~\de t\de s
\]
as an iterated stochastic integral 
\[
    \int_0^1 \int_0^1 \theta(t,s) ~\de B_t\de B_s.
\]
Since $(B_t-B_s)^2=\int_s^t \int_s^t \de B_u \de B_v$, Fubini implies these are equal for 
\[
    \theta(t,s)=2\min(s,t)\big(1-\max(s,t)\big).
\]
Next we find the eigenfunctions of the associated kernel $K$, i.e. those $v\in L^2([0,1])$ satisfying
\begin{align*}
    \frac{1}{2}\int_0^1 \theta(t,s) v(s) \de s
    &=
    \int_0^1 \min(s,t)\big(1-\max(s,t)\big) v(s)~\de s
    \\
    &=
    \int_0^t s(1-t) v(s)\de s
    +
    \int_t^1 t(1-s) v(s)\de s
    \\
    &=
    \lambda v(t)/2.
\end{align*}
In fact since $\theta$ is the Green's function for Brownian motion with Dirichlet boundary conditions on $[0,1]$, the orthonormal eigenfunctions are simply $v_k(t)=\sqrt{2}\sin(\pi k t)$ for $k\geq 1$, as can also be checked directly. Including the factor $\beta$, the eigenvalue $\lambda_k$ for $v_k$ is $\frac{2\beta}{\pi^2 k^2}$. Recalling ~\eqref{eq:resolvent-kernel}, we find
\[
    \tilde \theta(t,s)
    =
    \sum_{k\geq 1}
    \frac{2\sin(\pi kt)\sin(\pi ks)}
    {1+\frac{\pi^2 k^2}{2\beta} }.
\]
By Proposition~\ref{prop:shepp}, the process $\bbP_{[0,1]}^{\la Q_{\theta}\ra}$ has covariance 
\[
    \bbE^{\bbP_{[0,1]}^{\la Q_{\theta}\ra}}[B_t B_s]
    =
    \min(t,s)-
    \int_0^t \int_0^s
    \tilde \theta(u,v)~\de v~\de u
\]
and so in particular,
\[
    \bbE^{\bbP_{[0,1]}^{\la Q_{\theta}\ra}}[B_t^2]
    =
    t-
    2\sum_{k\geq 1} \frac{(1-\cos(\pi kt))^2}
    {\pi^2 k^2 \big(1+\frac{\pi^2 k^2}{2\beta}\big)}.
\]
To bound this variance, we first observe that it vanishes when we formally set $\beta=\infty$. (This is to be expected given the definition of $\wt\bbP_{0,\beta}$.) Indeed after expanding the numerators and using the cosine double angle formula, the identity
\begin{equation}
\label{eq:fourier-identity}
    \sum_{k\geq 1} \frac{(1-\cos(\pi kt))^2}
    {\pi^2 k^2}
    =
    t/2,\quad t\in [0,1]
\end{equation}
follows from the easily verified cosine Fourier expansion
\[
    \sum_{k\geq 1} \frac{\cos(\pi kt)}{\pi ^2 k^2}
    =
    \frac{t^2}{4}-\frac{|t|}{2}+\frac{1}{6},\quad\forall~t\in [-1,1].
\]
We conclude that uniformly over $t\in [0,1]$,
\begin{align*}
    \bbE^{\wt\bbP_{0,\beta}}[\|\bB_t\|^2/3]
    &=
     \bbE^{\bbP_{[0,1]}^{\la Q_{\theta}\ra}}[B_t^2]
    \\
    &=
    2\sum_{k\geq 1}
    \lt(
    \frac{(1-\cos(\pi kt))^2}
    {\pi^2 k^2}
    -
    \frac{(1-\cos(\pi kt))^2}
    {\pi^2 k^2 \big(1+\frac{\pi^2 k^2}{2\beta}\big)}
    \rt)
    \\
    &\leq
    \frac{2}{\pi^2}\sum_{k\geq 1}
    k^{-2}
    \cdot
    \lt(
    1-\frac{1}{1+\frac{\pi^2 k^2}{2\beta}}
    \rt)
    \\
    &\leq 
    \sum_{k\geq 1}
    \frac{1}{\beta+\frac{\pi^2 k^2}{2}}
    \\
    &\leq
    C\beta^{-1/2}
    .
    \qedhere
\end{align*}
\end{proof}

\section{Iteratively Improving the Confinement}
\label{sec:iterate}

In this section we improve our upper bounds on short-time fluctuations. The idea is that, for $i\in [T]$ corresponding to ``good'' intervals, we can recurse between improved local fluctuation bounds
\[
    \sup_{s,t\in [i,i+1]}\|\bB_s-\bB_t\|\leq R
\]
(aiming to decrease $R$) and improved Gaussian domination via $\wt\bbP^{(\eta)}_{i,\beta}$ (aiming to increase $\beta$). This strategy leads to the following recursion. \revedit{With $R_1=C_{\ref{lem:good-decomposition}}^2\sqrt{\log \alpha}$} and (for later convenience) $\beta_0=0$, we inductively define 
\revedit{
\begin{align}
\label{eq:beta-k}
    \beta_k&=\frac{\alpha}{16\text{e}^2R_k^3},
    \\
\label{eq:R-recursion}
    R_{k+1}
    &=
    C_{\ref{lem:good-decomposition}}
    C_{\ref{lem:unif-good}}
    (\log\beta_k)^{1/2} \beta_k^{-1/4}
    \asymp(\log\alpha)^{1/2}\lt(\frac{R_k^3}{\alpha}\rt)^{1/4}.
\end{align}
}
The latter equivalence holds \revedit{up to absolute constants} since $\beta_k$ increases up to $\Theta\lt(\frac{\alpha^4}{(\log\alpha)^6}\rt)$. In fact the recursion converges modulo constant factors in a small number of iterations.

\revedit{
\begin{lemma}
\label{lem:R-recursion}
    There exists $C_{\ref{lem:R-recursion}}$ such that for all $\alpha\geq 2$, for some $L\leq \revedit{C_{\ref{lem:R-recursion}}}\log \alpha $, we have
    \begin{align*}
    R_L&\leq \frac{C_{\ref{lem:R-recursion}}(\log\alpha)^2}{\alpha},
    \\
    \frac{C_{\ref{lem:R-recursion}}\alpha^4}{(\log\alpha)^6}
    \geq 
    \beta_L&\geq
    \frac{\alpha^4}{C_{\ref{lem:R-recursion}}(\log\alpha)^6}
    .
    \end{align*}
\end{lemma}

\begin{proof}
    The estimate $\frac{C_{\ref{lem:R-recursion}}\alpha^4}{(\log\alpha)^6}\geq \beta_L$ holds for all $L\geq 0$ by induction.
    For the rest, note that $R_{\ell+1}/R_{\ell}\leq 1/2$ holds until the first time $\ell$ such that $R_{\ell}\leq \frac{\revedit{C_{\ref{lem:R-recursion}}}(\log\alpha)^2}{\alpha}$. 
\end{proof}
}

\begin{remark}
\label{rem:effective-mass-general-p}
    For the more general bound \eqref{eq:effective-mass-general-p} with $0<p<2$, one just replaces \eqref{eq:beta-k}, \eqref{eq:R-recursion} by
    \begin{align*}
    \beta_{k,p}&=\frac{\alpha}{C_p R_{k,p}^{2+p}},
    \\
    R_{k+1,p}
    &=
    \revedit{C_{\ref{lem:good-decomposition}}
    C_{\ref{lem:unif-good}}}
    (\log\alpha)^{1/2}
    \alpha^{-1/4} R_{k,p}^{\frac{2+p}{4}}.
    \end{align*}
    This recursion improves until reaching
    \begin{align*}
    R_{L,p}&\asymp C_p' (\log\alpha)^{\frac{2}{2-p}}/\alpha^{\frac{1}{2-p}};
    \\
    \beta_{L,p}&\asymp \frac{\alpha^{\frac{4}{2-p}}}{C_p'(\log\alpha)^{\frac{4+2p}{2-p}}}.
    \end{align*}
    Up to constants, $\beta_{L,p}$ turns out to be the final lower bound we obtain for the effective mass as can be seen from the end of Section~\ref{sec:slow-oscillation}. The rest of the proofs adapt to general $p$ with infinitesimal modifications (e.g. one should replace occurrences of $\alpha^{-10}$ by $\alpha^{-C_p}$).
\end{remark}

\subsection{Inductive Decomposition of the Path Measure}
\label{subsec:induct}

Lemma~\ref{lem:general-path-decomp} generalizes Lemma~\ref{lem:good-decomposition}. Recall $Q_i(\bB)$ from \eqref{eq:Qibeta}.

\begin{lemma}
\label{lem:general-path-decomp}
    Fix $\alpha\geq 2$ and $k\in \bbZ_+$, and let $\eta>0$ be sufficiently small. There exist constants $c_1,c_2,\dots,c_k\in [0,1]$ and probability measures $\nu_{\good_0},\nu_{\good_1},\dots,\nu_{\good_k},\nu_{\bad_{1}},\dots,\nu_{\bad_{k}}$ on $C^{(\eta)}([0,1];\bbR^3)$ with $\nu_{\good_0}=\bbP^{(\eta)}_{[0,1]}$ such that for each $1\leq j\leq k$ there is a mixture decomposition
    \begin{equation}
    \label{eq:general-mixture-decomposition}
    \nu_{\good_{j-1}}
    =(1-c_{j})\nu_{\good_{j}}
    +
    c_{j}\nu_{\bad_{j}}.
    \end{equation}
    Moreover this decomposition has the following properties:
    \begin{enumerate}
    \item 
    \label{it:induct-1}
    Under the measure \revedit{$\nu_{\good_k}$}, the bound
    \revedit{$\sup_{s,t\in [0,1]}\|\bB_t-\bB_s\| \leq R_k$} holds almost surely for $R_k$ as in \eqref{eq:R-recursion}.
    \item 
    \label{it:induct-2}
    $\nu_{\good_k}\preceq \nu_{\good_{k-1}}\preceq\dots \preceq \nu_{\good_1}\preceq  \bbP^{(\eta)}_{[0,1]}$.
    \item 
    \label{it:induct-3}
    For each $1\leq j\leq k$, 
    \[
    (\nu_{\bad_j})^{\la\beta_{j-1} Q_0\ra} \preceq (\bbP^{(\eta)}_{[0,1]})^{\times 2}.
    \]
    \item 
    \label{it:induct-4}
    For each $1\leq j\leq k$, \eqref{eq:general-mixture-decomposition} is the reweighting of a decomposition
   \begin{equation}
    \label{eq:level-k-decomp-circ}
    \wt\bbP^{(\eta)}_{0,\beta_{j-1}}
    =
    (1-\delta^{\circ}_j)\nu^{\circ}_{\good_j}
    +
    \delta^{\circ}_j \nu^{\circ}_{\bad_j}.
    \end{equation}
    such that 
    \begin{equation}
    \label{eq:mu-circ-good-k-confined}
    \sup_{s,t\in [i,i+1]}\|\bB_t-\bB_s\| \leq R_j
    \end{equation}
    holds $\nu^{\circ}_{\good_j}$-almost surely, and
\begin{align}
\label{eq:delta-k-bound}
    \delta^{\circ}_j&\leq \alpha^{-10},
    \\
\label{eq:mu-circ-good-k-dominated}
    \nu^{\circ}_{\good_j}&\preceq \wt\bbP^{(\eta)}_{0,\beta_{j-1}},
    \\
\label{eq:mu-circ-bad-k-dominated}
    \nu^{\circ}_{\bad_j}&\preceq \big(\wt\bbP^{(\eta)}_{0,\beta_{j-1}}\big)^{\times 2}.
\end{align}
    \end{enumerate}
\end{lemma}

\begin{proof}
    We induct on $k$, where the base case is Lemma~\ref{lem:good-decomposition}. Consider a decomposition of the form \eqref{eq:level-k-decomp-circ} with $j=k$. By Lemmas~\ref{lem:unif-good} and \ref{lem:good-decomposition}, there exists such a decomposition such that all the properties in Hypothesis~\ref{it:induct-4} hold.
    \revedit{
    Indeed, Lemma~\ref{lem:unif-good} and the definition~\ref{eq:R-recursion} of $R_j$ ensure that the symmetric convex set
    \[
    K^{(\eta)}_{R_j/C_{\ref{lem:good-decomposition}}}
    =
    \lt\{
    \bB_{[0,1]}\in C^{(\eta)}([0,1];\bbR^3)~:~\sup_{s,t\in [0,1]}\|\bB_t-\bB_s\|\leq R_j/C_{\ref{lem:good-decomposition}}
    \rt\}
    \]
    satisfies the hypotheses of Lemma~\ref{lem:good-decomposition} with $\delta\leq\beta_{j-1}^{-100}\leq \alpha^{-10}$. (The latter bound is technically vacuous as stated for $\beta_0=0$, but this is the base case for which we use the construction following Lemma~\ref{lem:good-decomposition} directly; alternatively, Lemma~\ref{lem:unif-good} holds for all $\beta\geq 0$ if $\beta$ is replaced by $\beta+1$ in the right-hand side \eqref{eq:unif-good}.)
    }

    To close the induction for \eqref{eq:general-mixture-decomposition}, we reweight \eqref{eq:level-k-decomp-circ} by $\de \nu_{\good_{k-1}}/\de \wt\bbP^{(\eta)}_{0,\beta_{k-1}}$ and obtain
\begin{equation}
\label{eq:level-k-decomp}
    \nu_{\good_{k-1}}
    =
    (1-c_k)\nu_{\good_k}
    +
    c_k \nu_{\bad_k}.
\end{equation}

We now verify that the first three induction hypotheses continue to hold. Hypothesis~\ref{it:induct-1} follows from \eqref{eq:mu-circ-good-k-confined} and absolute continuity. For hypothesis~\ref{it:induct-2}, note that \eqref{eq:mu-circ-good-k-dominated} immediately implies $\nu_{\good_k}\preceq\nu_{\good_{k-1}}$ since the relation $\preceq$ is unchanged by joint reweighting. Finally hypothesis~\ref{it:induct-3} holds because 
\begin{align*}
    (\nu_{\bad_k})^{\la\beta_{k-1}Q_0\ra} 
    &\preceq
    \frac{\exp(-\beta_{k-1}Q_0)\de \nu_{\good_{k-1}}}{\de \wt\bbP^{(\eta)}_{0,\beta_{k-1}}}
    \cdot
    \big(\wt\bbP^{(\eta)}_{0,\beta_{k-1}}\big)^{\times 2}
    \\
    &\preceq
    \frac{\exp(-\beta_{k-1}Q_0)\de \bbP^{(\eta)}_{[0,1]}}{\de \wt\bbP^{(\eta)}_{0,\beta_{k-1}}}
    \cdot
    \big(\wt\bbP^{(\eta)}_{0,\beta_{k-1}}\big)^{\times 2}
    \\
    &=
    \big(\wt\bbP^{(\eta)}_{0,\beta_{k-1}}\big)^{\times 2}
    \\
    &
    \preceq
    (\bbP_{[0,1]}^{(\eta)})^{\times 2}.
\end{align*}
    \revedit{Here the first domination relation follows from \eqref{eq:mu-circ-bad-k-dominated} with $j=k$ since
    \[
    \frac{\de \nu_{\bad_k}}{\de\nu^{\circ}_{\bad_k}}
    =
    \frac{\de \nu_{\good_{k-1}}}{\de\wt\bbP^{(\eta)}_{0,\beta_{k-1}}}.
    \]
    Meanwhile $\nu_{\good_{k-1}}\preceq \bbP^{(\eta)}_{[0,1]}$ in the second step by hypothesis~\ref{it:induct-2},
    while the last step holds because the dilation $(\cdot)^{\times 2}$ preserves the relation $\preceq$.
    }
\end{proof}

Unwrapping the decompositions \eqref{eq:general-mixture-decomposition} of $\nu_{\good_j}$ in Lemma~\ref{lem:general-path-decomp}, we obtain
\begin{equation}
\label{eq:general-decomp-base}
    \bbP^{(\eta)}_{[0,1]}
    =a_{k}\nu_{\good_k}
    +
    \sum_{j=1}^k a_{k,j}\nu_{\bad_j};
\end{equation} 
\[
    a_k=\prod_{\ell\leq k} (1-c_{\ell}),\quad
    a_{k,j}=c_j\prod_{\ell=1}^{j-1} (1-c_{\ell}).
\]
As in \eqref{eq:product-decomp-basic} taking a $T$-fold product yields
\begin{align}
\label{eq:general-product-decomp}
    \bbP^{(\eta)}_{[0,T]}
    &=
    \sum_{\gamma^{(k)}\in 
    \{\good_k,\bad_k,\dots,\bad_1\}^T}
    w(\gamma^{(k)})
    \bP^{(\eta)}_{\gamma^{(k)}};
    \\
    \notag
    \bP^{(\eta)}_{\gamma^{(k)}}&\equiv\prod_{i=0}^{T-1} \nu_{i,\gamma^{(k)}_i},
    \\
    \notag
    w(\gamma^{(k)})&=
    a_k^{|\gamma_{\good_k}|}
    \cdot
    \prod_{j=1}^k
    a_{k,j}^{|\gamma_{\bad_j}|}
    .
\end{align}
\revedit{Here analogously to \eqref{eq:gamma-good-1-def} we write 
\[
\gamma_{\good_k}\equiv\{i\in [T]~:~\gamma_{i}=\good_k\},
\quad\quad\quad
\gamma_{\bad_j}\equiv\{i\in [T]~:~\gamma_{i}=\bad_j\}.
\]
}
Reweighting by $W_{\alpha,T}^{(A)}$, the resulting analog of \eqref{eq:product-decomp-polaron} is:
\begin{equation}
\label{eq:product-decomp-polaron-general}
\begin{aligned}
    \wh\bbP_{\alpha,T}^{(A,\eta)}
    &=
    \sum_{\gamma^{(k)}\in
    \{\good_k,\bad_k,\dots,\bad_1\}}
    \wh w(\gamma^{(k)})
    \wh\bP^{(A,\eta)}_{\gamma^{(k)}}
    ;
    \\
    \de\wh\bP^{(A,\eta)}_{\gamma^{(k)}}
    &\propto
    W_{\alpha,T}^{(A)}(\bB)
    ~\de\bP^{(\eta)}_{\gamma^{(k)}}(\bB);
    \\
    \wh w(\gamma^{(k)})
    &=
    w(\gamma^{(k)})\cdot 
    \frac{
    \bbE^{\bP^{(A,\eta)}_{\gamma^{(k)}}}\lt[W_{\alpha,T}^{(A)}\rt]
    }
    {
    \bbE^{\bbP^{(A,\eta)}_{\alpha,T}}\lt[W_{\alpha,T}^{(A)}\rt]
    }
.
\end{aligned}
\end{equation}

\begin{lemma}
\label{lem:general-mostly-good}
The weights $\wh w(\gamma^{(k)})$ in \eqref{eq:product-decomp-polaron-general} satisfy
    \begin{equation}
    \label{eq:general-mostly-good} 
     \sum_{\gamma^{(k)}\in 
     \{\good_k,\bad_k,\dots,\bad_1\}} 
     \wh w(\gamma^{(k)})\cdot|\gamma^{(k)}_{\good_k}|
     \geq
     T(1-k\alpha^{-10}).
    \end{equation}
\end{lemma}

\begin{proof}
    Fix $i\in [T]$ arbitrarily and define 
    \begin{equation}
    \label{eq:pij-def}
    p_{i,j}=
    \sum_{\gamma^{(j)}\in \{\good_j,\bad_j,\dots,\bad_1\}^T:~\gamma^{(j)}(i)=\good_j}
    \wh w(\gamma^{(j)})
    \end{equation}
    to be the probability that $\gamma^{(j)}(i)=\good_j$ according to the decomposition~\eqref{eq:product-decomp-polaron-general} at level $j$. 
    We will show below that $\frac{p_{i,j+1}}{p_{i,j}}\geq 1-\alpha^{-10}$ for each $j$. This suffices to complete the proof: it implies
    \[
    p_{i,j}\geq (1-\alpha^{-10})^k\geq 1-k\alpha^{-10}
    \]
    and summing over $i\in [T]$ then yields \eqref{eq:general-mostly-good}.
    
    Similarly to \eqref{eq:nu-extend}, we extend $\nu^{\circ}_{i,\good_{j+1}}$ and $\nu_{i,\good_{j+1}}$ and $\wt\bbP^{(\eta)}_{i,\beta_{j}}$ 
    to probability measures on $C^{(\eta)}([0,T];\bbR^3)$ rather than $C^{(\eta)}([i,i+1];\bbR^3)$ via
    \begin{equation}
    \label{eq:extend-v2}
    \begin{aligned}
    \nu^{\circ}_{i,\good_{j+1},[0,T]}
    &=
    \bbP^{(\eta)}_{[0,i]}\times \nu^{\circ}_{i,\good_{j+1}}\times \bbP^{(\eta)}_{[i+1,T]},
    \\
    \nu_{i,\good_{j+1},[0,T]}
    &=
    \bbP^{(\eta)}_{[0,i]}\times \nu_{i,\good_{j+1}}\times \bbP^{(\eta)}_{[i+1,T]},
    \\
     \wt\bbP^{(\eta)}_{i,\beta_{j},[0,T]}
    &=
    \bbP^{(\eta)}_{[0,i]}\times  \wt\bbP^{(\eta)}_{i,\beta_{j}}\times \bbP^{(\eta)}_{[i+1,T]}.
    \end{aligned}
    \end{equation}
    \revedit{
    The reason to define these measures is that e.g. $\nu_{i,\good_{j+1},[0,T]}$ is given by summing all these components in \eqref{eq:general-product-decomp} (for $k=j+1$) which satisfy $\gamma_i^{(j+1)}=\good_{j+1}$.
    Hence it appears naturally in a decomposition of $\bbP^{(\eta)}_{[0,T]}$ with weight given by the sum in \eqref{eq:pij-def}, but with $\wh w(\gamma^{(j+1)})$ replaced by $w(\gamma^{(j+1)})$. 
    We will prove \eqref{eq:pij-def} by comparing these different reweightings of the same decomposition, using as usual the Gausian correlation inequality.
    }

    To begin this argument, for each $i\in [T]$ and $j$ we have the reweighting of \eqref{eq:general-decomp-base} by $W^{(A)}_{\alpha,T}$:
    \begin{equation}
    \label{eq:pij-decomposition}
    \begin{aligned}
    \wh\bbP_{\alpha,T}^{(A,\eta)}
    &=
    p_{i,j} 
    (\nu_{i,\good_{j},[0,T]})^{(W^{(A)}_{\alpha,T})}
    +
    (p_{i,j-1}-p_{i,j})
    (\nu_{i,\bad_{j},[0,T]})^{(W^{(A)}_{\alpha,T})}+
    \\
    &\quad\quad\quad\quad
    \ldots
    +
    (1-p_{i,1})
    (\nu_{i,\bad_{1},[0,T]})^{(W^{(A)}_{\alpha,T})}.
    \end{aligned}
    \end{equation}
    Here the weight values follow by the definition \eqref{eq:pij-def} and Proposition~\ref{prop:refine-partition}. In using the latter, each $S_i$ corresponds as just discussed to a decomposition of $\nu_{i,\good_{j},[0,T]}$ (resp.  $\nu_{i,\bad_{\ell},[0,T]}$) into the parts $\bP^{(\eta)}_{\gamma^{(k)}}$ with $\gamma^{(k)}(i)=\good_j$ (resp. $\gamma^{(k)}(i)=\bad_{\ell}$).

    To estimate $\frac{p_{i,j+1}}{p_{i,j}}$, first note that we can extend \eqref{eq:level-k-decomp-circ} (with $j+1$ in place of $j$) to a decomposition 
    \begin{equation}
    \label{eq:intermediate-step-1}
    \wt\bbP^{(\eta)}_{i,\beta_j,[0,T]}
    =
    (1-\delta^{\circ}_{j+1})\nu^{\circ}_{i,\good_{j+1},[0,T]}
    +
    \delta^{\circ}_{j+1} \nu^{\circ}_{i,\bad_{j+1},[0,T]}
    \end{equation}
    by taking products with $\bbP^{(\eta)}_{[0,i]}\times \bbP^{(\eta)}_{[i+1,T]}$ as in \eqref{eq:extend-v2}. By definition, reweighting \eqref{eq:intermediate-step-1} by $\exp(\beta_j Q_i(\bB))$ transforms the left-hand side into $\bbP^{(\eta)}_{i,[0,T]}$. Hence reweighting \eqref{eq:intermediate-step-1} by $\exp(\beta_j Q_i(\bB))
    \cdot\frac{\de\nu_{i,\good_{j},[0,T]}}{\de\bbP^{(\eta)}_{[0,T]}}$ transforms the left-hand side into 
    $\nu_{i,\good_{j},[0,T]}$. (Note that reweighting can be done in multiple stages as discussed just before Proposition~\ref{prop:refine-partition}.)

    We claim that reweighting \eqref{eq:intermediate-step-1} by  
    \[
    p_{i,j} W_{\alpha,T}^{(A)}(\bB)\exp(\beta_j Q_i(\bB))
    \cdot\frac{\de\nu_{i,\good_{j},[0,T]}}{\de\bbP^{(\eta)}_{[0,T]}}
    \]
    yields the decomposition (of subprobability measures):
    \begin{equation}
    \label{eq:intermediate-step-2}
    p_{i,j}(\nu_{i,\good_j,[0,T]})^{(W^{(A)}_{\alpha,T})}
    =
    p_{i,j+1}(\nu_{i,\good_{j+1},[0,T]})^{(W^{(A)}_{\alpha,T})}
    +
    (p_{i,j}-p_{i,j+1})(\nu_{i,\bad_{j+1},[0,T]})^{(W^{(A)}_{\alpha,T})}
    .
    \end{equation}
    The form of the left-hand side in \eqref{eq:intermediate-step-2} follows from the previous paragraph. 
    Moreover, this left-hand side is the contribution of $(\nu_{i,\good_j,[0,T]})^{(W^{(A)}_{\alpha,T})}$ to $\wh\bbP^{(A,\eta)}_{\alpha,T}$ in \eqref{eq:pij-decomposition}.
    Replacing $j$ by $j+1$ in \eqref{eq:pij-decomposition}, we obtain the right-hand side weights in \eqref{eq:intermediate-step-2}.

    Given this reweighting from \eqref{eq:intermediate-step-1} to \eqref{eq:intermediate-step-2}, the formula \eqref{eq:reweight3} thus implies that
    \begin{equation}
    \label{eq:pj-ratio-new}
    \frac{p_{i,j+1}}{p_{i,j}}
    =
    (1-\delta^{\circ}_j)
    \cdot
    \frac{
    \bbE^{\nu^{\circ}_{i,\good_{j+1},[0,T]}}\lt[W_{\alpha,T}^{(A)}(\bB)\exp(\beta_j Q_i(\bB))
    \cdot\frac{\de\nu_{i,\good_{j},[0,T]}}{\de\bbP^{(\eta)}_{[0,T]}}\rt]
    }
    {
    \bbE^{\wt\bbP^{(\eta)}_{i,\beta_{j},[0,T]}}\lt[W_{\alpha,T}^{(A)}(\bB)\exp(\beta_j Q_i(\bB))
    \cdot\frac{\de\nu_{i,\good_{j},[0,T]}}{\de\bbP^{(\eta)}_{[0,T]}}\rt]
    }.
    \end{equation}
    Recalling that $\delta^{\circ}_j\leq\alpha^{-10}$, to show $\frac{p_{i,j+1}}{p_{i,j}}\geq 1-\alpha^{-10}$ it therefore remains to prove that
    \begin{equation}
    \label{eq:hard-GCI-use}
    \begin{aligned}
    &\bbE^{\nu^{\circ}_{i,\good_{j+1},[0,T]}}\lt[W_{\alpha,T}^{(A)}(\bB)\exp(\beta_j Q_i(\bB))
    \cdot\frac{\de\nu_{i,\good_{j},[0,T]}}{\de\bbP^{(\eta)}_{[0,T]}}\rt]
    \\
    &\stackrel{?}{\geq}
    \bbE^{\wt\bbP^{(\eta)}_{i,\beta_{j},[0,T]}}\lt[W_{\alpha,T}^{(A)}(\bB)\exp(\beta_j Q_i(\bB))
    \cdot\frac{\de\nu_{i,\good_{j},[0,T]}}{\de\bbP^{(\eta)}_{[0,T]}}\rt].
    \end{aligned}
    \end{equation}
    This will follow by the Gaussian correlation inequality. \revedit{Indeed since $\bbE^{\bP}[f]=\bbE^{\bQ}\lt[\frac{\de \bP}{\de\bQ} \cdot f\rt]$ for mutually absolutely continuous probability measures $\bP$ and $\bQ$, the left-hand side of \eqref{eq:hard-GCI-use} equals}
    \begin{equation}
    \label{eq:LHS-equals}
    \bbE^{\wt\bbP^{(\eta)}_{i,\beta_{j},[0,T]}}
    \lt[
    \lt(\frac{\de\nu^{\circ}_{i,\good_{j+1},[0,T]}}{\de\wt\bbP^{(\eta)}_{i,\beta_{j},[0,T]}}
    \rt)
    \cdot 
    \lt(W_{\alpha,T}^{(A)}(\bB)\exp(\beta_j Q_i(\bB))
    \cdot\frac{\de\nu_{i,\good_{j},[0,T]}}{\de\bbP^{(\eta)}_{[0,T]}}
    \rt)\rt]
    \end{equation}
    The former Radon--Nikodym derivative is symmetric-quasi-concave by \eqref{eq:mu-circ-good-k-dominated}. Because $\nu_{i,\good_{j},[0,T]}$ is supported in $K^{(\eta)}_{R_j,i}$ (recall Lemma~\ref{lem:general-path-decomp} part~\ref{it:induct-1}), the \revedit{remainding factor in \eqref{eq:LHS-equals}} equals
    \[
    \lt(
    W_{\alpha,T}^{(A)}(\bB)\exp(\beta_j Q_i(\bB))\cdot \one_{K^{(\eta)}_{R_j,i}}
    \rt)
    \cdot
    \lt(
    \frac{\de\nu_{i,\good_{j},[0,T]}}{\de\bbP^{(\eta)}_{[0,T]}}
    \rt).
    \]
    \revedit{Thanks to the factor $\one_{K^{(\eta)}_{R_j,i}}$, both of these terms are limits of products of symmetric-quasi-concave functions} (the first by Proposition~\ref{prop:gaussian-domination-general}).
    Moreover $\wt\bbP^{(\eta)}_{i,\beta_{j},[0,T]}$ is a centered Gaussian measure.
    Therefore we may indeed apply the Gaussian correlation inequality to lower-bound \eqref{eq:LHS-equals} by
    \begin{align*}
    &\bbE^{\wt\bbP^{(\eta)}_{i,\beta_{j},[0,T]}}
    \lt[
    \frac{\de\nu^{\circ}_{i,\good_{j+1},[0,T]}}{\de\wt\bbP^{(\eta)}_{i,\beta_{j},[0,T]}}
    \rt]
    \cdot 
    \bbE^{\wt\bbP^{(\eta)}_{i,\beta_{j},[0,T]}}
    \lt[W_{\alpha,T}^{(A)}(\bB)\exp(\beta_j Q_i(\bB))
    \cdot\frac{\de\nu_{i,\good_{j},[0,T]}}{\de\bbP^{(\eta)}_{[0,T]}}
    \rt]
    \\
    &\hspace{4.4cm}
    =
    \bbE^{\wt\bbP^{(\eta)}_{i,\beta_{j},[0,T]}}
    \lt[W_{\alpha,T}^{(A)}(\bB)\exp(\beta_j Q_i(\bB))
    \cdot\frac{\de\nu_{i,\good_{j},[0,T]}}{\de\bbP^{(\eta)}_{[0,T]}}
    \rt].
    \end{align*}
    We conclude that \eqref{eq:hard-GCI-use} holds, which completes the proof.
\end{proof}

The definition~\eqref{eq:dominating-product} extends similarly to \eqref{eq:product-decomp-polaron-general}. However for use in the next section we will now include adjacent interval interactions in the dominating Gaussian measure. We define
\begin{align*}
    \wt\bbP^{(\eta)}_{i,\bad_j}
    &=
    (\bbP^{(\eta)}_{[i,i+1]})^{\times 2},
    \\
    \wt\bbP^{(\eta)}_{i,\good_k}  
    &=
    \wt\bbP^{(\eta)}_{i,\beta_k}.
\end{align*}
Further, let $S^{\circ}(\gamma^{(k)})$ denote the set of $i\in [T]$ such that $\{i,i+1\}\subseteq \gamma^{(k)}_{\good_k}$. For any subset $S\subseteq S^{\circ}(\gamma^{(k)})$, we set
\begin{equation}
\label{eq:wt-bP-gamma-S-def}
    \wt\bP^{(\eta)}_{\gamma^{(k)},S}
    \propto
    \exp\lt(
    -\beta_k
    \sum_{i\in S}
    \int_{i}^{i+1}
    \int_{i+1}^{i+2}
    \|\bB_t-\bB_s\|^2
    ~\de t~\de s
    \rt)
    \cdot
    \prod_{i\in [T]}
    \wt\bbP^{(\eta)}_{i,\gamma^{(k)}_i}
    .
\end{equation}
Then we have the following analog of Lemma~\ref{lem:gaussian-domination-one-step}.

\begin{lemma}
\label{lem:gaussian-domination-k-step}
    For each $\gamma^{(k)}\in \{\good_k,\bad_k,\dots,\bad_1\}^T$ and any $S\subseteq S^{\circ}(\gamma^{(k)})$, we have
    \[
    \wh\bP^{(A,\eta)}_{\gamma^{(k)}}\preceq \wt\bP^{(\eta)}_{\gamma^{(k)},S}.
    \]
    In particular,
    \[
    \bbE^{ \wh\bP^{(A,\eta)}_{\gamma^{(k)}}}[\|\bB_T\|^2]
    \leq
    \bbE^{ \wt\bP^{(\eta)}_{\gamma^{(k)},S}}[\|\bB_T\|^2].
    \]
\end{lemma}

\begin{proof}
 \revedit{We will show that Proposition~\ref{prop:gaussian-domination-general} applies}. Set $\bbQ^{(\eta)}_{i,\good_k}=\bbP^{(\eta)}_{[i,i+1]}$ and $\bbQ^{(\eta)}_{i,\bad_j}=(\bbP^{(\eta)}_{[i,i+1]})^{\times 2}$. Also define $\bbQ^{\dagger,(\eta)}_{i,\good_k}=\lt(\nu_{i,\good_k}\rt)^{\la \beta_k Q_i\ra}$ and $\bbQ^{\dagger,(\eta)}_{i,\bad_j}=\lt(\nu_{i,\bad_j}\rt)^{\la \beta_{j-1} Q_i\ra}$. Then we use:
    \begin{align*}
    \bbQ
    &=
    \bQ^{(\eta)}_{\gamma^{(k)}}
    \equiv
    \prod_{i\in [T]}\bbQ^{(\eta)}_{i,\gamma^{(k)}_i},
    \\
    \bbQ^{\dagger}&=
    \prod_{i\in[T]} 
    \bbQ^{\dagger,(\eta)}_{i,\gamma^{(k)}_i}
    ,
    \\
    \wt\bbQ
    &=
    \wt\bP_{\gamma^{(k)}}^{(\eta)},
    \\
    \wh\bbQ
    &=
    \wh\bP_{\gamma^{(k)}}^{(A,\eta)}.
    \end{align*}
    Lemma~\ref{lem:general-path-decomp} implies that 
    \[
    \bbQ^{\dagger,(\eta)}_{i,\gamma^{(k)}_i}\preceq \bbQ^{(\eta)}_{i,\gamma_i}
    \]
    for each $i$. Hence taking a product over $i\in [T]$ yields $\bbQ^{\dagger}\preceq\bbQ$. The Radon--Nikodym derivative $\de\wh\bbQ/\de\bbQ^{\dagger}$ is proportional to $W^{(A)}_{\alpha,T}$ as required. Meanwhile $\de\wt\bbQ/\de\bbQ$ takes the required form with
    \[
    F(s,t)=
    \begin{cases}
    \beta_{j-1},\quad s,t\in [i,i+1],~\gamma_i=\bad_j;
    \\
    \beta_k,\quad s,t\in [i,i+1],~\gamma_i=\good_k;
    \\
    \beta_k,\quad (s,t)\in [i,i+1]\times [i+1,i+2],~i\in S;
    \\
    0,\quad\text{else}.
    \end{cases}
    \]
    \revedit{
    Indeed since $\beta_k=\frac{\alpha}{2\text{e}^2 (2R_k)^3}$ (recall \eqref{eq:beta-k}), the condition \eqref{eq:F-bound} applies on each interval $[i,i+2]$ via Lemma~\ref{lem:general-path-decomp}, part~\ref{it:induct-1} (with slightly adapted constants).
    }
    Thus Proposition~\ref{prop:gaussian-domination-general} yields the claim.
\end{proof}

\section{Slow Oscillation on Long Time-Scales}
\label{sec:slow-oscillation}

For $\bB\in C([0,T];\bbR^3)$ and $i\in [T]$, define the interval average
\[
    \overline{\bB}_{[i,i+1]}=\int_i^{i+1} \bB_t~\de t.
\]
Moreover for $0\leq a<b\leq T$, define the probability measure
\[
    \overline{\bbP}^{(\eta)}_{[a,b],\beta}
    \propto
     \exp\lt(
     -\beta
    \sum_{i=a}^{b-2}
    \int_{i}^{i+1}
    \int_{i+1}^{i+2}
    \|\bB_t-\bB_s\|^2
    ~\de t~\de s
    \rt)
    \cdot
    \prod_{i=a}^{b-1} \wt\bbP^{(\eta)}_{i,\beta}.
\]
This is essentially the $[a,b]$ factor in \eqref{eq:wt-bP-gamma-S-def} if $S$ contains a contiguous block $\{a,a+1,\dots,b-2\}\subseteq S$ (and also $a-1,b-1\notin S$). In particular, note that the quadratic confining interactions occur both within individual intervals and between adjacent intervals. We show below that the values $\overline{\bB}_{[i,i+1]}$ have smaller increments under $\overline{\bbP}^{(\eta)}_{[a,b],\beta}$ than suggested by the single-value fluctuation bounds of Lemma~\ref{lem:gain}. In the following lemmas we continue to treat $\beta\geq 2$ as arbitrary, but will set $\beta=\beta_L\asymp \frac{\alpha^4}{(\log\alpha)^6}$ when using them to finally deduce Theorem~\ref{thm:main}.

\begin{lemma}
\label{lem:small-averaged-fluctuations}
    \[
    \bbE^{\overline{\bbP}^{(\eta)}_{[0,2],\beta}}
    \lt[
    \|\overline{\bB}_{[0,1]}
    - 
    \overline{\bB}_{[1,2]}\|^2
    \rt]
    \leq O(1/\beta).
    \]
\end{lemma}

\begin{proof}
    Consider the quadratic forms
    \begin{align*}
        Q_{[0,2]}(\bB)&=\int_0^2 \int_0^2 \|\bB_t-\bB_s\|^2~\de t~\de s;
        \\
        Q_{[0,1],[1,2]}(\bB)&=\|\overline{\bB}_{[0,1]}-\overline{\bB}_{[1,2]}\|^2.
    \end{align*}
    \revedit{
    Note that $Q_{[0,2]}(\bB)\geq 2Q_{[0,1],[1,2]}(\bB)\geq Q_{[0,1],[1,2]}(\bB)$ for any $\bB$ since
    \begin{align*}
    Q_{[0,2]}(\bB)-2Q_{[0,1],[1,2]}(\bB)
    &=
    \int_0^{1} \int_0^{1} \|\bB_t-\bB_s\|^2~\de t~\de s
    +
    \int_{1}^2 \int_{1}^2 \|\bB_t-\bB_s\|^2~\de t~\de s
    \\
    &\quad\quad
    +
    \int_0^{1} \|\bB_t-\overline{\bB}_{[0,1]}\|^2 ~\de t
    +
    \int_{1}^2 \|\bB_t-\overline{\bB}_{[1,2]}\|^2 ~\de t.
    \end{align*}
    }
    In light of Corollary~\ref{cor:quadratic-domination}, it suffices to show the result with $\overline{\bbP}^{(\eta)}_{[0,2],\beta}$ replaced by $\bbP^{\dagger,(\eta)}_{[0,2],\beta}(\bB)$ where for general $i$,
    \begin{equation}
    \label{eq:difference-reweight}
    \de\bbP^{\dagger,(\eta)}_{[i,i+2],\beta}(\bB)
    \propto
    \exp\lt(
    -
    \beta\cdot \|\overline{\bB}_{[i,i+1]}-\overline{\bB}_{[i+1,i+2]}\|^2 
    \rt)
    \de\bbP_{[i,i+2]}^{(\eta)}(\bB).
    \end{equation}
    Finally, observe that the distribution of the difference $\overline{\bB}_{[0,1]}-\overline{\bB}_{[1,2]}$ is already a centered Gaussian under $\bbP^{(\eta)}_{[0,2]}$. Therefore its reweighted distribution under $\bbP^{\dagger,(\eta)}_{[0,2],\beta}(\bB)$ is also a centered Gaussian and has variance $O(1/\beta)$ as desired.
\end{proof}

\begin{lemma}
\label{lem:small-summed-fluctuations}
    For $s\geq 2\ell$,
    \[
    \bbE^{\overline{\bbP}^{(\eta)}_{[0,s],\beta}}
    \lt[
    \lt\|
    \sum_{i=0}^{\ell-1}
    \overline{\bB}_{[2i,2i+1]}
    - 
    \overline{\bB}_{[2i+1,2i+2]}
    \rt\|^2
    \rt]
    \leq O(\ell/\beta).
    \]
\end{lemma}

\begin{proof}
    Similarly to the previous proof, Corollary~\ref{cor:quadratic-domination} implies that it suffices to show the same estimate under the product measure (recall \eqref{eq:difference-reweight})
    \[
    \bbP^{\dagger,(\eta)}_{[0,2\ell],\beta}
    \equiv
    \prod_{i=0}^{\ell-1}\bbP^{\dagger,(\eta)}_{[2i,2i+2],\beta}.
    \]
    Lemma~\ref{lem:small-averaged-fluctuations} finishes the proof as $\lt\{\overline{\bB}_{[2i,2i+1]}- \overline{\bB}_{[2i+1,2i+2]}\rt\}_{0\leq i\leq \ell-1}$ are $\bbP^{\dagger,(\eta)}_{[0,2\ell],\beta}$-independent.
\end{proof}

\begin{lemma}
\label{lem:small-shifted-summed-fluctuations}
    For $s\geq 2\ell-1$,
    \[
    \bbE^{\overline{\bbP}^{(\eta)}_{[0,s],\beta}}
    \lt[
    \lt\|
    \sum_{i=1}^{\ell-1}
    \overline{\bB}_{[2i-1,2i]}
    - 
    \overline{\bB}_{[2i,2i+1]}
    \rt\|^2
    \rt]
    \leq O(\ell/\beta).
    \]
\end{lemma}

\begin{proof}
    By Corollary~\ref{cor:quadratic-domination}, we can discard interactions outside $[1,2\ell-1]$ to obtain
    \[
    \overline{\bbP}^{(\eta)}_{[0,s],\beta}
    \preceq
    \bbP^{(\eta)}_{[0,1]}
    \times
    \overline{\bbP}^{(\eta)}_{[1,2\ell-1],\beta}
    \times
    \bbP^{(\eta)}_{[2\ell-1,s]}.
    \]
    Thus it suffices to bound the expectation under the right-hand measure, which is equivalent to Lemma~\ref{lem:small-summed-fluctuations}.
\end{proof}

\begin{lemma}
\label{lem:start-to-end-change-smoothed}
    \[
    \bbE^{\overline{\bbP}^{(\eta)}_{[0,s],\beta}}
    \lt[
    \lt\|
    \overline{\bB}_{[0,1]} 
    - 
    \overline{\bB}_{[s-1,s]} 
    \rt\|^2
    \rt]
    \leq O(s/\beta).
    \]
\end{lemma}

\begin{proof}
    For $s=2\ell$ this is immediate from Lemmas~\ref{lem:small-summed-fluctuations} and \ref{lem:small-shifted-summed-fluctuations} via the identity
    \[
    \overline{\bB}_{[0,1]} 
    - 
    \overline{\bB}_{[2\ell-1,2\ell]} 
    =
    \lt(
    \sum_{i=0}^{\ell-1}
    \overline{\bB}_{[2i,2i+1]}
    - 
    \overline{\bB}_{[2i+1,2i+2]}
    \rt)
    +
    \lt(
    \sum_{i=1}^{\ell-1}
    \overline{\bB}_{[2i-1,2i]}
    - 
    \overline{\bB}_{[2i,2i+1]}
    \rt)
    \]
    and the fact that \revedit{$(a+b)^2\leq 2(a^2+b^2)$}. The case of $s$ odd is analogous.
\end{proof}

\begin{lemma}
\label{lem:start-to-end-change}
    \[
    \bbE^{\overline{\bbP}^{(\eta)}_{[0,s],\beta}}
    \lt[
    \lt\|
    \bB_{s}-\bB_0
    \rt\|^2
    \rt]
    \leq O\lt(\frac{s}{\beta}+\beta^{-1/2}\rt).
    \]
\end{lemma}

\begin{proof}
    By Lemma~\ref{lem:gain}, Corollary~\ref{cor:quadratic-domination} and Jensen's inequality,
    \[
    \bbE^{\overline{\bbP}^{(\eta)}_{[0,s],\beta}}
    \lt[
    \lt\|
    \bB_0-\overline{\bB}_{[0,1]}
    \rt\|^2
    \rt]
    \leq
    \sup_{u\in [0,1]}
    \bbE^{\overline{\bbP}^{(\eta)}_{[0,1],\beta}}
    \lt[
    \lt\|
    \bB_0-\bB_u
    \rt\|^2
    \rt]
    \leq
    O(1/\beta^{1/2}).
    \]
    Similarly, 
    \[
    \bbE^{\overline{\bbP}^{(\eta)}_{[0,s],\beta}}
    \lt[
    \lt\|
    \bB_s-\overline{\bB}_{[s-1,s]}
    \rt\|^2
    \rt]
    \leq
    O(1/\beta^{1/2}).
    \]
    Combining these with Lemma~\ref{lem:start-to-end-change-smoothed} implies the result via 
    \revedit{$(a+b+c)^2\leq 3(a^2+b^2+c^2)$}.
\end{proof}

\begin{proof}[Proof of Theorem~\ref{thm:main}]
    As mentioned at the start of Subsection~\ref{subsec:outline} we assume $T\geq 1$ is actually an integer; if not, one can just rescale time by a constant factor or equivalently discretize time in non-integer increments of $T/\lfloor T\rfloor$. 
    Recalling Lemma~\ref{lem:R-recursion}, 
    \revedit{
    we take $L\leq C_{\ref{lem:R-recursion}}\log \alpha$ such that
    \begin{equation}
    \label{eq:beta-L-final}
    \beta_L\geq \frac{\alpha^4}{C_{\ref{lem:R-recursion}}(\log\alpha)^6}.
    \end{equation}
    }
    Fix $\gamma^{(L)}\in \{\good_L,\bad_L,\dots,\bad_1\}^T$. If $i\notin \gamma^{(L)}_{\good_L}$, call the interval $[i,i+1]$ \emph{bad}. Call each bad interval a \emph{bad block}, and each connected component of good intervals a \emph{good block}. \revedit{Let $b(\gamma^{(L)})$ be the number of bad blocks.}

    Next, let $S(\gamma^{(L)})\subseteq S^{\circ}(\gamma^{(L)})\subseteq [T]$ consist of all $i$ such that both $[i,i+1]$ and $[i+1,i+2]$ are good intervals. We apply Lemma~\ref{lem:gaussian-domination-k-step} with this choice to obtain
    \begin{equation}
    \label{eq:blocks-bound-final}
    \bbE^{ \wh\bP^{(A,\eta)}_{\gamma^{(L)}}}[\|\bB_T\|^2]
    \leq
    \bbE^{ \wt\bP^{(\eta)}_{\gamma^{(L)},S(\gamma^{(L)})}}[\|\bB_T\|^2].
    \end{equation}
    Moreover $\wt\bP^{(\eta)}_{\gamma^{(L)},S(\gamma^{(L)})}$ has independent increments on distinct blocks, so the right-hand side of \eqref{eq:blocks-bound-final} is given by a sum over these blocks. Its law on each bad block $[i,i+1]$ is $(\bbP^{(\eta)}_{[i,i+1]})^{\times 2}$ which contributes $O(1)$. Meanwhile on a good block $[a,a+s]$ of length $s$, its law is $\overline{\bbP}^{(\eta)}_{[a,a+s],\beta_L}$ and
    \[
    \bbE^{\overline{\bbP}^{(\eta)}_{[a,a+s],\beta_L}}
    \lt[\|\bB_{a+s}-\bB_a\|^2\rt]
    \leq
    O\lt(\frac{s}{\beta_L}+\frac{1}{\sqrt{\beta_L}}\rt)
    \] 
    by Lemma~\ref{lem:start-to-end-change}. Since there are $b(\gamma^{(L)})$ bad blocks, there are at most $b(\gamma^{(L)})+1$ good blocks. 
    Summing contributions yields
    \begin{equation}
    \label{eq:sum-over-blocks}
    \bbE^{ \wt\bP^{(\eta)}_{\gamma^{(L)},S(\gamma^{(L)})}}[\|\bB_T\|^2]
    \leq
    O\lt(\frac{T}{\beta_L}+b(\gamma^{(L)})+\beta_L^{-1/2}\rt).
    \end{equation}
    Averaging over $\gamma^{(L)}$ via \eqref{eq:product-decomp-polaron-general}, we find 
\begin{align*}
    \bbE^{\wh\bbP_{\alpha,T}^{(A,\eta)}}[\|\bB_T\|^2]
    &=
    \sum_{\gamma^{(L)}\in
    \{\good_L,\bad_L,\dots,\bad_1\}}
    \wh w(\gamma^{(L)})
    \bbE^{\wh\bP^{(A,\eta)}_{\gamma^{(L)}}}[\|\bB_T\|^2]
    \\
    &\leq
    \sum_{\gamma^{(L)}\in
    \{\good_L,\bad_L,\dots,\bad_1\}}
    \wh w(\gamma^{(L)})
    \bbE^{\wt\bP^{(\eta)}_{\gamma^{(L)},S(\gamma^{(L)})}}
    [\|\bB_T\|^2]
    \\
    &\stackrel{\eqref{eq:sum-over-blocks}}
    {\leq}
    O\lt(
    \frac{T}{\beta_L}
    +
    \sum_{\gamma^{(L)}\in
    \{\good_L,\bad_L,\dots,\bad_1\}}
    \wh w(\gamma^{(L)})
    \big(b(\gamma^{(L)})+\beta_L^{-1/2}\big)
    \rt)\\
    &\stackrel{Lem.~\ref{lem:general-mostly-good}}{\leq}
    O\lt(\frac{T}{\beta_L} + 
    L\alpha^{-10}T
    +\beta_L^{-1/2}
    \rt)
    \\
    &
    \leq O\lt(\frac{T}{\beta_L} 
    +\beta_L^{-1/2}
    \rt).
\end{align*}
    The last line follows since $L\alpha^{-10}\leq O(\alpha^{-9})\leq \beta_L^{-1}$ \revedit{by Lemma~\ref{lem:R-recursion}}.
    Recalling \eqref{eq:beta-L-final}, we have established \eqref{eq:desired-bound-after-simplification} which concludes the proof of Theorem~\ref{thm:main}.
\end{proof}

\section*{Acknowledgement}

Thanks to Ramon van Handel, Bo'az Klartag, Andrea Montanari, Ron Peled, and Scott Sheffield for helpful discussions. We were introduced to the polaron by excellent lectures of Erwin Bolthausen and S.R.S. Varadhan during Ofer Zeitouni's 60th birthday conference at NYU. We are also grateful to Krzysztof My{\'s}liwy for pointing us to \cite{mysliwy2022polaron} and its connection to Remark~\ref{rem:robust}, as well as to Volker Betz and Tobias Schmidt for suggesting a correction and several clarifications in Section~\ref{sec:iterate}.
Finally we thank the anonymous referees for many other helpful suggestions.
\bibliographystyle{alpha}
\bibliography{bib}

\end{document}